\documentclass[sigconf]{acmart}
\renewcommand\footnotetextcopyrightpermission[1]{} 

\usepackage[linesnumbered,ruled,vlined]{algorithm2e}
\usepackage{balance}
\usepackage{latexsym}
\usepackage{amsfonts}

\usepackage{amsmath}
\DeclareMathOperator*{\argmin}{arg\,min}

\usepackage{amssymb}
\usepackage{eufrak} 
\usepackage{mathrsfs} 
\usepackage{color}
\usepackage{colortbl}
\usepackage{epsfig}
\usepackage{xspace}
\usepackage{graphicx}
\usepackage{subfigure}
\usepackage{enumerate}
\usepackage{booktabs}
\usepackage{pdfpages}
\usepackage{epsfig}
\usepackage{multirow}
\usepackage{url}
\usepackage{soul}
\usepackage{wrapfig}
\usepackage{thmtools,thm-restate}
\usepackage{caption}
\usepackage{verbatim}
\usepackage{listings}

\newcommand{\eg}{e.\,g.,\@}
\newcommand{\ie}{i.\,e.,\@}
\newcommand{\wrt}{w.\,r.\,t.\@}
\newcommand{\at}[1]{\protect\ensuremath{\mathsf{#1}}\xspace}

\newenvironment{problem1}{\begin{em}
  \refstepcounter{problem}{\vspace{1ex} \noindent \textsc{Minimum Conversion Tree Problem.}}}{
\end{em}} 

\newenvironment{problem2}{\begin{em}
		\refstepcounter{problem}
		{\vspace{1ex} \noindent \textsc{Plan Enumeration Problem.}}}{
\end{em}} 

\newcommand{\myparagraph}[1]{\vspace{0.1cm}\noindent\textbf{{#1}.~}}

\newcounter{enum}
\newenvironment{packed_enum}{
\begin{list}{\textbf{(\arabic{enum})}}{
  \setlength{\itemsep}{0pt}
  \setlength{\parskip}{0pt}
  \setlength{\labelwidth}{-5 pt}
  \setlength{\leftmargin}{0 pt}
  \setlength{\itemindent}{0pt}
  \usecounter{enum}}
}{\end{list}}

\definecolor{javared}{rgb}{0.6,0,0} 
\definecolor{javagreen}{rgb}{0.25,0.5,0.35} 
\definecolor{javapurple}{rgb}{0.5,0,0.35} 
\definecolor{javadocblue}{rgb}{0.25,0.35,0.75} 
\lstdefinelanguage{scala}{
	morekeywords={abstract,case,catch,class,def,%
		do,else,extends,false,final,finally,%
		for,if,implicit,import,match,mixin,%
		new,null,object,override,package,%
		private,protected,requires,return,sealed,%
		super,this,throw,trait,true,try,%
		type,val,var,while,with,yield},
	otherkeywords={=>,<-,<\%,<:,>:,\#,@},
	sensitive=true,
	morecomment=[l]{//},
	morecomment=[n]{/*}{*/},
	morestring=[b]",
	morestring=[b]',
	morestring=[b]"""
}
\let\oldnl\nl
\newcommand{\nonl}{\renewcommand{\nl}{\let\nl\oldnl}}

\makeatletter
    \newcommand\figcaption{\def\@captype{figure}\caption}
    \newcommand\tabcaption{\def\@captype{table}\caption}
\makeatother

\newcommand{\add}[1]{{\textcolor{black}{{#1}}}}

\usepackage{hyperref}

\begin{document}

\setcopyright{none}

\fancyhead{}
\settopmatter{printacmref=false, printfolios=false}

\newcommand{\rheem}{\textsc{Rheem}\xspace}
\newcommand{\rheemx}{the \textsc{SystemX} optimizer\xspace}
\newcommand{\Rheemx}{The \textsc{SystemX} optimizer\xspace}
\newcommand{\open}{boundary\xspace}

\newcommand{\task}[1]{\protect{\texttt{#1}}\xspace}
\newcommand{\pl}[1]{\protect\ensuremath{\mathsf{#1}}\xspace}
\newcommand{\ds}[1]{\protect{{\em #1}}\xspace}

\newcommand{\rememberlines}{\xdef\rememberedlines{\number\value{AlgoLine}}}
\newcommand{\resumenumbering}{\setcounter{AlgoLine}{\rememberedlines}}

\newcommand{\placetextbox}[3]{
	\setbox0=\hbox{#3}
	\AddToShipoutPictureFG*{
		\put(\LenToUnit{#1\paperwidth},\LenToUnit{#2\paperheight}){\vtop{{\null}\makebox[0pt][c]{#3}}}%
	}%
}%

\title{Optimizing Applications for Cross-Platform Execution}
\title{RHEEMix in the Data Jungle: \\A Cost-based Optimizer for Cross-Platform Systems}

\author{
	Sebastian Kruse
}\authornote{Work partially done while interning at QCRI.}
\affiliation{
	\institution{Hasso Plattner Institute \\University of Potsdam}
	\country{Germany}
}
\author{Zoi Kaoudi}
\affiliation{
	\institution{Qatar Computing Research Institute\\ Hamad Bin Khalifa University}
	\country{Qatar}
}

\author{Bertty Contreras}
\affiliation{
	\institution{Qatar Computing Research Institute\\ Hamad Bin Khalifa University}
	\country{Qatar}
}

\author{Sanjay Chawla}
\affiliation{
	\institution{Qatar Computing Research Institute\\ Hamad Bin Khalifa University}
	\country{Qatar}
}

\author{Felix Naumann}
\affiliation{
	\institution{Hasso Plattner Institute \\University of Potsdam}
	\country{Germany}
}

\author{Jorge-Arnulfo Quian\'e-Ruiz}
\affiliation{
	\institution{Qatar Computing Research Institute\\ Hamad Bin Khalifa University}
	\country{Qatar}
}

\placetextbox{0.5}{0.97}{\color{blue}\fbox{Please refer to the latest version published at VLDB Journal 2020 https://doi.org/10.1007/s00778-020-00612-x}}%

\begin{abstract}
In pursuit of efficient and scalable data analytics, the insight that ``one size does not fit all'' has given rise to a plethora of specialized data processing platforms
and today's complex data analytics are moving beyond the limits of a single platform.
In this paper, we present the cost-based optimizer of Rheem, an open-source cross-platform system that copes with these new requirements.
The optimizer allocates the subtasks of data analytic tasks to the most suitable platforms.
Our main contributions are:
(i)~a mechanism based on graph transformations to explore alternative execution strategies;
(ii)~a novel graph-based approach to efficiently plan data movement among subtasks and platforms; and
(iii)~an efficient plan enumeration algorithm, based on a novel enumeration algebra.
We extensively evaluate our optimizer under diverse real tasks.
The results show that our optimizer is capable of selecting the most efficient platform combination for a given task, freeing data analysts from the need to choose and orchestrate platforms.
In addition, our optimizer allows tasks to run more than one order of magnitude faster by using multiple platforms instead of a single platform.
\end{abstract}

\maketitle

\section{Cross-Platform Data Processing}
\label{section:intro}

Modern data analytics are characterized by
(i)~increasing query/task\footnote{\small{Henceforth, we use the term task without loss of generality.}} complexity,
(ii)~heterogeneity of data sources, and
(iii)~a proliferation of data processing platforms (\emph{platforms}, for short).
As a result, today's data analytics often need to perform {\em cross-platform data processing}, \ie~running their tasks on more than one platform.
The research and industry communities have recently identified this need~\cite{stonebraker2015the,rheem-vision-edbt16} and have proposed systems to support different aspects of cross-platform data processing~\cite{bigdawg-demo,apacheBeam,apacheCalcite,rheem-system-vldb18,ires-bigdata,gog2015musketeer}.

The current practice to cope with cross-platform requirements is to write ad-hoc programs to glue different specialized platforms together~\cite{duggan2015bigdawg,apacheDrill,luigi,apacheBeam,apacheCalcite}.
This approach is not only expensive and error-prone, but it also requires to know the intricacies of the different platforms to achieve high efficiency.
Thus, there is an urgent need for a systematic solution that enables efficient cross-platform data processing without requiring users to specify which platform to use.
The holy grail would be to replicate the success of DBMSs to cross-platform applications: users formulate platform-agnostic data analytic tasks and an intermediate system decides on which platforms to execute each (sub)task with the goal of minimizing cost (\eg~runtime or monetary cost).
Recent research works have taken first steps towards that direction~\cite{SimitsisWCD12,gog2015musketeer,ires-bigdata,myria-cidr17}.
Nonetheless, they all lack important aspects.
For instance, none of these works consider different alternatives for data movement, even though
having a single way to move data from one platform to another (\eg~via files) may hinder cross-platform opportunities.
Additionally, most focus on specific applications, such as ETL~\cite{SimitsisWCD12} or specific platforms~\cite{polybase,miso}.
Recently, commercial engines, such as DB2 and Teradata, have extended their systems to support different platforms 
but none provides a systematic solution, \ie~users have to specify the platform to use.

The key component for such a systematic solution is  a {\em cross-platform optimizer} and it is thus the focus of this paper.
Concretely, we consider the problem of {\em finding the set of platforms that minimizes the execution cost of a given task}.
One might think of a rule-based optimizer: 
\eg~execute a task on a centralized/distributed platform when the input data is small/large.
However, this approach is neither practical nor effective.
First, setting rules at the task level implicitly assumes that all the operations in a task have the same computational complexity and input cardinality.
Such assumptions do not hold in practice, though.
Second, the cost of a task on any given platform depends on many input parameters, which hampers a rule-based optimizer's effectiveness as it oversimplifies the problem.
Third, as new platforms and applications emerge, maintaining a rule-based optimizer becomes very cumbersome.
We thus pursue a {\em cost-based} approach.

Devising a cost-based optimizer for cross-platform settings is quite challenging for many reasons:
(i)~the optimization search space grows exponentially with the number of atomic operations of the given data analytic task;
(ii)~platforms vastly differ \wrt~their supported operations and processing abstractions;
(iii)~the optimizer must consider the cost of moving data across platforms;
(iv) cross-platform settings are characterized by high uncertainty, \ie~data distributions are typically unknown and cost functions are hard to calibrate; and
(iv)~the optimizer must be extensible to accommodate new platforms and emerging application requirements.

In this paper, we delve into the cross-platform optimizer of \rheem~\cite{rheem-demo-sigmod16,rheem-system-vldb18}, our open source cross-platform system~\cite{rheemwebsite}.
To the best of our knowledge, our optimizer is the first to tackle all of the above challenges.
The idea is to split a single task into multiple atomic operators and to find the most suitable platform for each operator (or set of operators) so that its cost is minimized.
After giving an overview on our optimizer (Section~\ref{section:overview}), we present our major contributions:
\begin{packed_enum}
\item
We propose a plan inflation mechanism that is a very compact representation of the entire plan search space
and provide a cost model purely based on UDFs
(Section~\ref{section:enrichment}).

\item We model data movement for cross-platform optimization as a new graph problem, which we prove to be NP-hard, and propose an efficient algorithm to solve it (Section~\ref{section:datamovement}).

\item We devise a new algebra and a new lossless pruning technique to enumerate executable cross-platform plans for a given task in a highly efficient manner (Section~\ref{section:enumeration}).

\item We discuss how we exploit our optimization pipeline for performing progressive optimization in order to deal with poor cardinality estimates (Section~\ref{section:progressive}).

\item We extensively evaluate our optimizer under diverse tasks using real-world datasets and show that it allows tasks to run more than one order of magnitude faster by using multiple platforms instead of a single platform (Section~\ref{section:experiments}).
\end{packed_enum}

Eventually, we discuss related work (Section~\ref{section:relatedwork}) and conclude this paper with a summary (Section~\ref{section:conclusion}).

\section{Background and Overview}
\label{section:overview}

\myparagraph{\rheem background}
\rheem decouples applications from platforms in order to enable cross-platform data processing.
Although decoupling data processing was the driving motive when designing \rheem, we also adopted a three-layer decoupled optimization approach, as envisioned in~\cite{rheem-vision-edbt16}.
One can see this three-layer optimization as a separation of concerns for query optimization.
Overall, as \rheem applications have good knowledge of the tasks' logic and the data they operate on, they are in charge of any logical, such as operator re-ordering (the application optimization layer). 
\rheem receives from applications an optimized procedural plan, for which it determines the best platforms for execution (the core optimization layer).
Then, the selected platforms run the plan by performing further physical platform-specific optimizations, such as setting the data buffer sizes (the platform optimization layer).
\rheem is at the core optimization layer.

%
\rheem is composed of two main components (among others): the {\em cross-platform optimizer} and the {\em executor}.
The cross-platform optimizer gets as input a {\em \rheem plan} and produces an {\em execution plan} by specifying the platform to use for each operator in the \rheem plan.
In turn, the executor orchestrates and monitors the execution of the generated execution plan on the selected platforms.
In this paper, we focus on the cross-platform optimizer.
Below, we first detail what \rheem and execution plans are and then give an overview of our cross-platform optimizer.


\myparagraph{\rheem plan}
As stated above, the input to \rheem optimizer is a procedural \rheem plan, which is essentially a directed data flow graph.
The vertices are {\em \rheem operators} and the edges represent the data flow among the operators.
Only \at{Loop} operators accept feedback edges, thus enabling iterative data flows.
A \rheem plan without any loop operator is essentially an acyclic graph.
Conceptually, the data is flowing from source operators through the graph and is manipulated in the operators until it reaches a sink operator.
\rheem operators are platform-agnostic and define a particular kind of data transformation over their input, \eg~a \at{Reduce} operator aggregates all input data into a single output.

\begin{figure}[!h]
	\vspace{-0.1cm}
	\centering
	\includegraphics[width=0.9 \columnwidth]{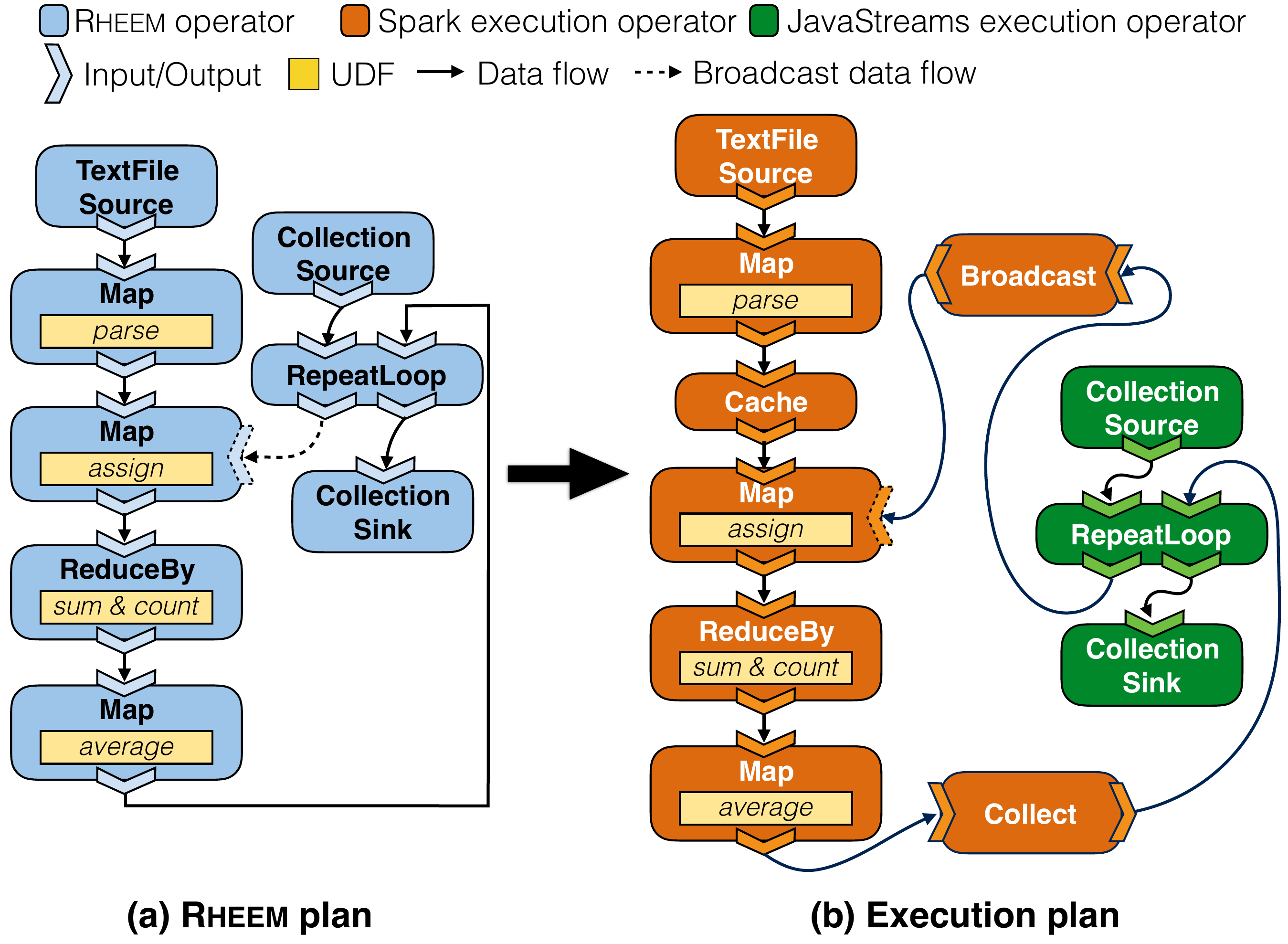}
	\vspace{-0.4cm}
	\caption{K-means example.}
	\label{figure:plan}
	\vspace{-0.3cm}
\end{figure}

\begin{example}
Figure~\ref{figure:plan}(a) shows a \rheem plan for k-means.
Data points are read via a \at{TextFileSource} and parsed using a \at{Map}, while the initial centroids are read via a \at{CollectionSource}.
The main operations of k-means (\ie~assigning the closest centroid to each data point and computing the new centroids) are repeated until convergence (\ie~the termination condition of \at{RepeatLoop}).
The resulting centroids are output in a collection.
For a tangible picture of the context in which our optimizer works, we point the interested reader to the examples of our source code\footnote{\url{https://github.com/rheem-ecosystem/rheem-benchmark}}.
\end{example}

\myparagraph{Execution plan}
Similarly to a \rheem plan, an execution plan is a data flow graph, but with two differences.
First, the vertices are (platform-specific) {\em execution operators}.
Second, the execution plan may comprise additional execution operators for data movement among platforms (\eg~a \at{Broadcast} operator).
Conceptually, given a \rheem plan, an execution plan indicates on which platform the executor must enact each \rheem operator.

\begin{example}
Figure~\ref{figure:plan}(b) shows the execution plan for the k-means \rheem plan when \pl{Spark}~\cite{apacheSpark} and Java Streams~\cite{javaStreams} are the only available platforms.
This plan exploits \pl{Spark}'s high parallelism for the large input dataset and at the same time benefits from the low latency of \pl{JavaStreams} for the small collection of centroids.
Also note the three additional execution operators for data movement (\at{Broadcast}, \at{Collect}) and to make data reusable (\at{Cache}).
As we show in Section~\ref{section:experiments}, such hybrid execution plans often achieve higher performance than plans with only a single platform.
\end{example}

\myparagraph{Cross-platform optimizer}
Recall that any logical or physical optimization takes place in the application layer.
Thus,
unlike traditional relational database optimizers, our cross-platform optimizer does not aim at finding operator orderings.
Instead, the goal of our optimizer is to select one or more platforms to execute a given \rheem plan in the most efficient manner.
The main idea is to split a single task into multiple atomic operators and to find the most suitable platform for each operator (or set of operators) so that the total cost is minimized.

\begin{figure}[!h]
	\vspace{-0.2cm}
	\centering
	\includegraphics[width=\columnwidth]{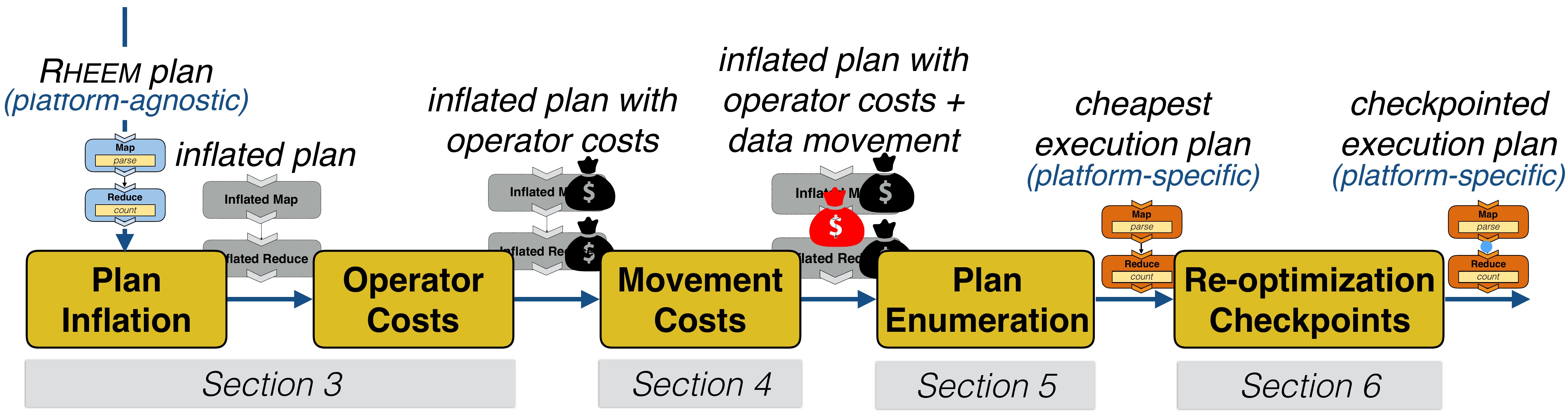}
	\vspace{-0.7cm}
	\caption{Cross-platform optimization pipeline.}
	\label{figure:optimizer}
	\vspace{-0.3cm}
\end{figure}

Figure~\ref{figure:optimizer} depicts the workflow of our optimizer.
At first, given a \rheem plan, the optimizer passes the plan through a plan enrichment phase (Section~\ref{section:enrichment}).
In this phase, the optimizer first {\em inflates} the input plan by applying a set of mappings.
These mappings determine how each of the platform-agnostic \rheem operators can be implemented on the different platforms with execution operators.
The result is an \emph{inflated \rheem plan} that can be traversed through alternative routes.
That is, the nodes of the resulting inflated plan are \rheem operators with all its execution alternatives.
Additionally, the optimizer {\em annotates} the inflated plan with estimates for both the intermediate result sizes and the costs of executing each execution operator.
Then, the optimizer takes a graph-based approach to determine how data can be moved most efficiently among execution operators of different platforms and again annotates the results and their costs to the plan (Section~\ref{section:datamovement}).
Finally, it uses all these annotations to determine the optimal execution plan via an enumeration algorithm (Section~\ref{section:enumeration}).
This algorithm is centered around an enumeration algebra and a highly effective, yet lossless pruning technique.
Eventually, the resulting execution plan can be enacted by \rheem cross-platform executor.

Additionally, as data cardinalities might be imprecise because of the uncertainty in cross-platform settings, \rheem monitors the real execution of the execution plan.
If the real cardinalities do not match the estimated ones, the executor pauses the execution of the plan and sends the subplan of still non-executed operators to the optimizer (Section~\ref{section:progressive}).
In return, the executor gets the newly optmized execution plan for the given subplan and resumes the execution.
We detail each of the above phases in the following four sections.

\myparagraph{Extensible design}
Note that the design of our optimizer allows for extensibility: adding a new platform to \rheem does not require any change to the optimizer codebase.
A developer has to simply provide new execution operators and their mappings to \rheem operators.
\rheem comes with default cost functions for the execution operators depending on their type.
Yet, the developer can use a profiling tool provided by \rheem to get her own cost functions for better cost estimates.


\section{Plan Enrichment}
\label{section:enrichment}

When our optimizer receives a \rheem plan, it has to do some preparatory work before it can explore alternative execution plans.
We refer to this phase as \emph{plan enrichment}.
Concretely, our optimizer
(i)~determines all eligible platform-specific execution operators for each \rheem operator (Section~\ref{section:enrichment_inflation}); and
(ii)~estimates the execution costs for these execution operators (Section~\ref{section:costestimation}).

\subsection{Inflation}
\label{section:enrichment_inflation}

While \rheem operators declare certain data processing operations, they do not provide an implementation and are thus not executable.
Therefore, our optimizer \emph{inflates} the \rheem plan with \add{all} corresponding execution operators, each providing an actual implementation on a specific platform.
A basic approach to determine corresponding execution operators for \rheem operators are mapping dictionaries, such as in~\cite{engineindependence,ires-bigdata}.
This approach allows only for $1$-to-$1$ \emph{operator mappings} between \rheem operators and execution operators -- more complex mappings are precluded, though.
However, different data processing platforms work with different abstractions:
While databases employ relational operators and Hadoop-like systems build upon \texttt{Map} and \texttt{Reduce}, special purpose systems (\eg~graph processing systems) rather provide specialized operators (\eg~for the \textsf{PageRank} algorithm).
Due to this diversity, 1-to-1 mappings are often insufficient and a flexible operator mapping technique is called for.

\myparagraph{Graph-based operator mappings}
To this end, we define operator mappings in terms of \emph{graph mappings},
which, in simple terms, map a matched subgraph to a substitute subgraph.
We formally define an operator mapping as follows.


\begin{definition}[Operator mapping]
An operator mapping $p \rightarrow s$ consists of the graph pattern $p$ and the substitution function $s$.
Assume that $p$ matches the subgraph $G$ of a given \rheem plan.
Then, the operator mapping designates the substitute subgraph $G' := s(G)$ for this match via the substitution function.
\end{definition}

Usually, the matched subgraph $G$ is a constellation of \rheem operators and the substitute subgraph $G'$ is a corresponding constellation of execution operators.
However, $G$ may comprise execution operators; and $G'$ may be a constellation of \rheem operators. The latter cases allow our optimizer to be able to choose among platforms that do not natively support certain operators. 
Figure~\ref{figure:enrichment}(a) exemplifies some mappings for our k-means example.

\begin{figure}[t!]
	\centering
	\includegraphics[scale=0.17]{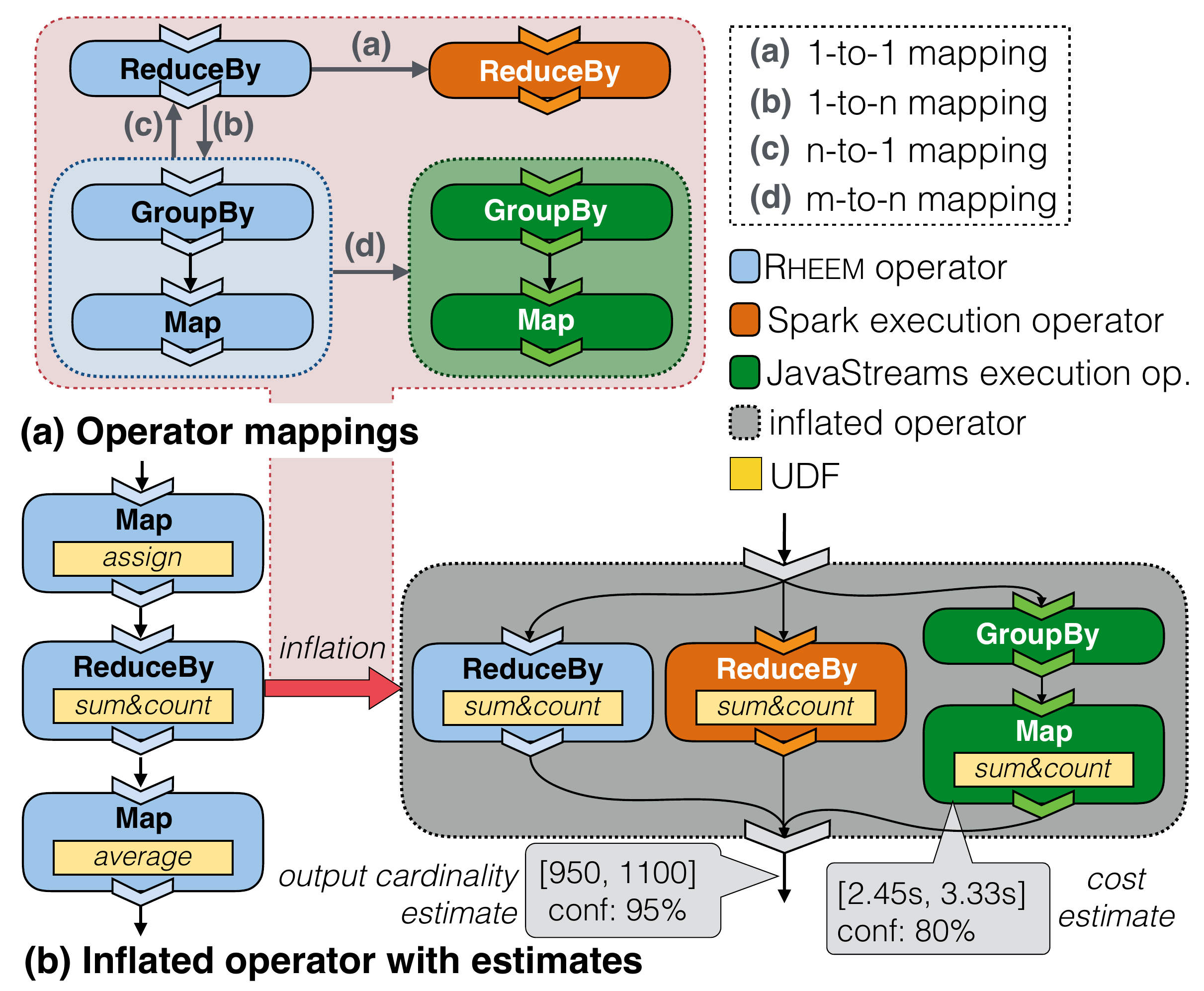}
	\vspace{-0.3cm}
	\caption{\rheem plan enrichment.}
	\label{figure:enrichment}
	\vspace{-0.2cm}
\end{figure}

\begin{example}[1-to-n mapping]
In Figure~\ref{figure:enrichment}(a), the 1-to-n mapping transforms the \at{ReduceBy} \rheem operator to a constellation of \at{GroupBy} and \at{Map} \rheem operators%
, which in turn are transformed to Java Streams execution operators.
\end{example}

In contrast to 
$1$-to-$1$ mapping approaches, our graph-based approach provides a more powerful means to derive execution operators from \rheem operators.
Our  approach also allows us to break down complex operators (\eg~\at{PageRank}) and map it to platforms that do not support it natively.
\add{Mappings are provided by developers when adding new \rheem or execution operators.}

\myparagraph{Inflated operator}
It is important to note that, \add{during the inflation phase,} our optimizer does not apply operator mappings by simply replacing matched subgraphs $G$ by \add{one of} their substitute subgraphs $G'$ as doing so would cause two insufficiencies:
First, this strategy would always create only a single execution plan, thereby precluding any cost-based optimization.
Second, that execution plan would be dependent on the order in which the mappings are applied, because once a mapping is applied, other relevant mappings might become inapplicable.
We overcome both insufficiencies by introducing \emph{inflated operators} in \rheem plans.
An inflated operator replaces a matched subgraph and comprises that matched subgraph \emph{and all} the substitute graphs.
The original subgraph is retained so that operator mappings can be applied \emph{in any order}; and each inflated operator can contain multiple substitute graphs, thereby accounting for \emph{alternative} operator mappings.
Ultimately, an inflated operator expresses alternative subplans inside \rheem plans.
\add{Thus, our graph-based mappings do not determine which platform to use for each \rheem operator but list all the alternatives for the optimizer to choose from.
This is in contrast to Musketeer~\cite{gog2015musketeer} and Myria~\cite{myria-cidr17}, which use their rewrite rules to directly obtain the platform that each operator should run on.}

\begin{example}[Operator inflation]
Consider again our k-means example whose plan contains a \at{ReduceBy} operator.
Figure~\ref{figure:enrichment}(b) depicts the inflation of that operator.
Concretely, the \rheem \at{ReduceBy} operator is replaced by an inflated operator that hosts both the original and two substitute subgraphs.
\end{example}

After our optimizer has exhaustively applied all its mappings,
the resulting {\em inflated \rheem plan} defines \emph{all} possible combinations of execution operators of the original \rheem plan -- but without \emph{explicitly} materializing them.
In other words, an inflated \rheem plan is a highly compact representation of all execution plans.



\subsection{Operators Cost Estimation}
\label{section:costestimation}
Once a \rheem plan is inflated, the optimizer estimates and annotates costs to each execution operator (see Figure~\ref{figure:enrichment})  by traversing the plan in a bottom-up fashion.
Cardinality and cost estimation are extremely challenging problems -- even in highly cohesive systems, such as relational databases, which have detailed knowledge on execution operator internals and data statistics~\cite{howgoodoptimizersare}.
As \rheem has little control on the underlying platforms, the optimizer uses a modular and fully {\em UDF-based cost model}.
Furthermore, it represents cost estimates as intervals with a confidence value, which allows it to perform on-the-fly re-optimization (Section~\ref{section:progressive}).

\begin{figure}[t!]
	\centering
	\includegraphics[scale=0.16]{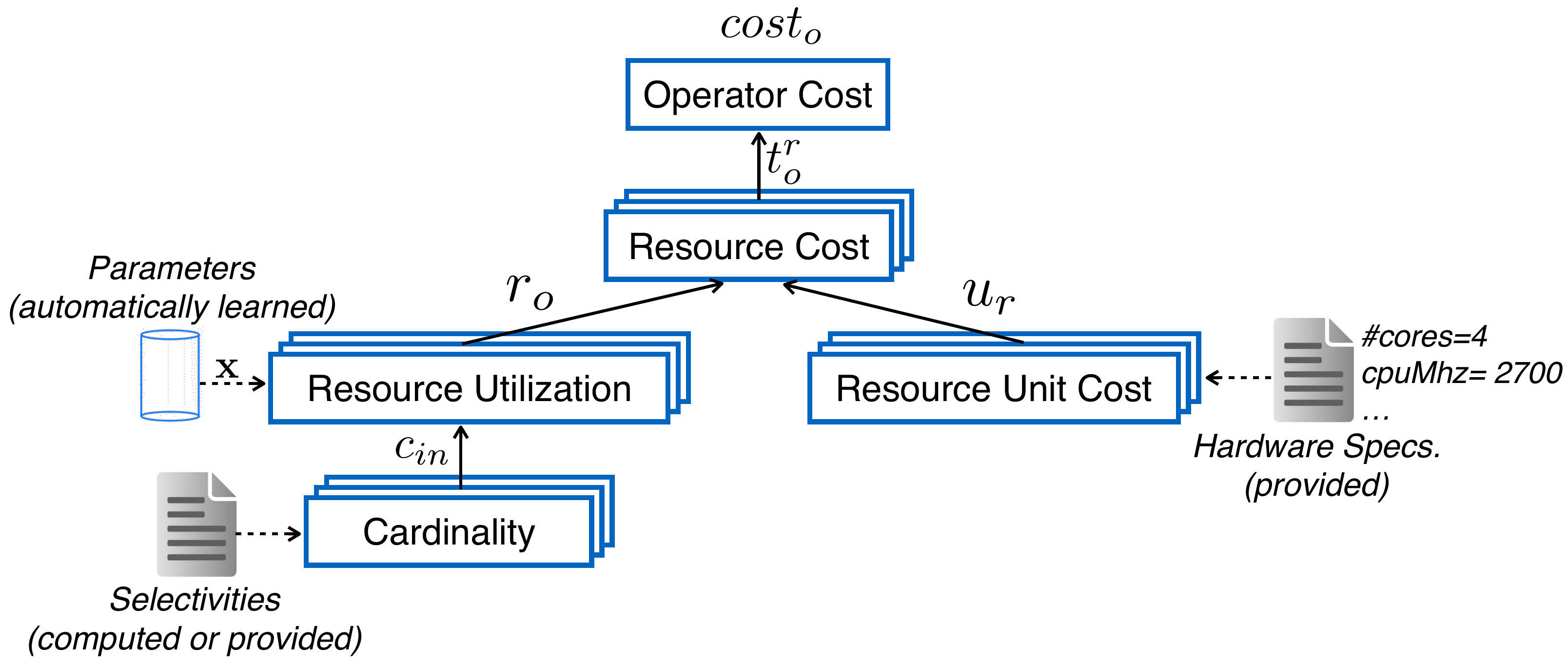}
	\vspace{-0.1cm}
	\caption{Operator cost estimation process.}
	\label{figure:costmodel}
	\vspace{-0.3cm}
\end{figure}

\myparagraph{Cost estimation}
We propose a simple, yet powerful approach that decouples the cost formulas to enable developers to intervene at any level of the cost estimation process.
This approach also allows the optimizer to be portable across different deployments.
Figure~\ref{figure:costmodel} illustrates this cost estimation process, where the boxes represent all the UDFs in the process.
The total cost estimate for an execution operator $o$ depends on the cost of the resources it consumes (CPU, memory, disk, and network), defined as: $cost_o = t_o^{CPU} + t_o^{mem} + t_o^{disk} + t_o^{net}$.
The cost of each resource $t_o^r$ is the product of (i)~its utilization $r_o$ and (ii)~the unit costs $u_r$ (\eg~how much one CPU cycle costs).
The latter depends on hardware characteristics (such as number of nodes and CPU cores), which are encoded in a configuration file for each platform.
On the other hand, the resource utilization is estimated by a cost function $r_o$ that depends on the input cardinality $c_{in}$ of its corresponding \rheem operator. 
For instance, the cost function to estimate the CPU cycles required by the \at{SparkFilter} operator is $CPU_{SF} := c_{in}(Filter) \times \alpha + \beta$, where parameter $\alpha$ denotes the number of required CPU cycles for each input data quantum and parameter $\beta$ describes some fixed overhead for the operator start-up and scheduling.
Notice that $cost_o$ contains the parameters of all the resources.
Obtaining the right values for these parameters, such as the $\alpha, \beta$ values, is very time-consuming if it is done manually via profiling.
For this reason, \rheem provides an {\em offline} cost learner module that uses historical execution logs in order to {\em learn} these parameters.
We model the cost as a regression problem.  The estimated execution time is $t' = \sum_{i}cost_i({\bf x}, c_i)$ where ${\bf x}$ is a vector with all the parameters that we need to learn, and $c_i$ is the input cardinalities. 
Let $t$ be the real execution time, we then seek ${\bf x}$ that minimizes the difference between $t$ and $t'$:  ${\bf x}_\text{min}=\argmin_{\bf x}\ loss(t, t' )$.
We consider a \emph{relative} loss function defined as: $loss(t, t')= \left(\frac{|t - t'| + s}{t + s}\right)^2$, where $s$ is a regularizer inspired by additive smoothing that tempers the loss for small $t$.
We then use a genetic algorithm~\cite{mitchell1998introduction} to find ${\bf x}_\text{min}$. 
Further discussing the cost learner is out of the scope of this paper, please refer to~\cite{rheem-system-vldb18} for more details.

\myparagraph{Cardinality estimation}
To estimate the output cardinality of each \rheem operator, the optimizer first computes the output cardinalities of the source operators via sampling
and then traverses the inflated plan in a bottom-up fashion. 
For this, each \rheem operator is associated with a cardinality estimator function, which considers its properties (\eg~selectivity and number of iterations) and input cardinalities.
For example, the \at{Filter} operator uses $c_{out}$(\at{Filter})$ := c_{in}$(\at{Filter})$ \times \sigma_f$, where $\sigma_f$ is the selectivity of the user's \at{Filter} operator.
To address the uncertainty inherent to the selectivity estimation the optimizer expresses the cardinality estimates in an interval with a confidence value.
Basically, this confidence value gives the likelihood that the interval indeed contains the actual cost value.
For the selectivities the optimizer relies on basic statistics such as the number of output tuples and of distinct values. 
These statistics can be provided by the application/developer or obtained by runtime profiling, similar to~\cite{hueske2012opening,sofa}.
If not available, the optimizer uses default values for the selectivities and relies on re-optimization for correcting the execution plan if necessary, similar to~\cite{db2-federated}.
Note that we intentionally do not consider devising a sophisticated mechanism for cardinality estimation as it is an orthogonal problem~\cite{selinger1979access} that has been studied independently over the years.
This allows us to study the effectiveness of our optimization techniques without interference from cardinality estimation.

\section{Data Movement}
\label{section:datamovement}

Selecting optimal platforms for an execution plan might require to move data across platforms and transform them appropriately for the target platform.
This leads to an inherent trade-off between choosing the optimal execution operators and minimizing data movement and transformation costs.
Our optimizer must properly explore this trade-off to find the overall optimal execution plan.

However, planning and assessing communication is challenging for various reasons.
\add{First, there might be several alternative data movement strategies, \eg~from RDD to a file or to a Java object.
One might think about transferring data between two platforms by serializing and deserializing data via a formatted file, such as in~\cite{gog2015musketeer,myria-cidr17}.
While this simple strategy is feasible, it is not always the most efficient one.
In fact, having only a file-based data movement strategy may lead to missing many opportunities for cross-platform data processing.
Second, the costs of each strategy must be assessed so that our optimizer can explore the trade-off between selecting optimal execution operators and minimizing data movement costs.
Considering the costs of different data movement strategies is also crucial for finding cross-platform opportunities.
Third, data movement might involve several intermediate steps to connect two operators of different platforms.}


To address these challenges, we represent the space of possible communication steps as a \emph{channel conversion graph}~(Section~\ref{section:channelconversion}).
This graph representation allows us to model the problem of finding the most efficient communication path among execution operators as a new graph problem: the {\em minimum conversion tree} problem~(Section~\ref{section:minimumtree}).
We devise a novel algorithm to efficiently solve this graph problem~(Section~\ref{section:algorithm}).

\subsection{Channel Conversion Graph}
\label{section:channelconversion}

\newcommand{\ccg}{CCG\xspace}
\newcommand{\Ccg}{CCG\xspace}
\newcommand{\ccgs}{CCGs\xspace}
\newcommand{\Ccgs}{CCGs\xspace}

The channel conversion graph (\emph{CCG} for short) is a graph whose vertices are data structure types (\eg~an RDD in Spark) and whose edges express conversions from one data structure to another.
Before formally defining the \ccg, let us first explain how we model data structures ({\em communication channels}) and data transformation ({\em conversion operators}).

\myparagraph{Communication channel}
Data can flow among operators via communication channels (or simply \emph{channels}), which form the vertices in the \ccg.
This can be for instance an internal data structure or stream within a data processing platform, or simply a file.
For example, the yellow boxes in Figure~\ref{figure:channel-conversion-graph} depict the standard communication channels considered by our optimizer for Java Streams and Spark.
Note that communication channels can be {\em reusable}, \ie~they can be consumed multiple times, or non-reusable, \ie~once they are consumed they cannot be used anymore.
For instance, a file is reusable, while a data stream is usually not.

\myparagraph{Conversion operator}
In certain situations, it becomes necessary to convert channels from one type to another, \eg~it might be necessary to convert an RDD to a \add{CSV file}.
Such conversions are handled by conversion operators, which form the edges in the \ccg.
Conversion operators are in fact regular execution operators:
For example, \rheem provides the \add{\textsf{Spark\-RDD\-To\-CSV-File}} operator, which simply reads the RDD and writes it to a \add{CSV file}.
Intuitively, the associated communication costs are incurred neither by the RDD nor the file but by the conversion operator.
Thus, given a cardinality estimate of the data to be moved, the optimizer computes the conversion costs as regular execution operator costs.

\myparagraph{Channel conversion graph}
We can now integrate communication channels and their conversions in a graph.

\begin{definition} [Channel conversion graph]
A {\em \ccg} is a directed graph $G:=(C, E, \lambda)$, where the set of vertices $C$ contains the channels, $E$~comprises the directed edges indicating that the source channel can be converted to the target channel, and $\lambda\colon~E~\rightarrow~O$ is a labeling function that attaches the appropriate conversion operator $o\in O$ to each edge $e\in E$.
\end{definition}

\rheem provides the \ccg with generic channels, \eg~CSV files, together with the channels of the supported platforms, \eg~RDDs.
Developers can easily extend the \ccg if needed, \eg~when adding a new platform to \rheem.
To this end, they exploit the fact that conversion operators are existing execution operators.
Thus, they simply provide mappings from the new channels to the existing ones by using the existing execution operators (often the source and sink operators which are for reading and writing data).


\subsection{Minimum Conversion Tree Problem}
\label{section:minimumtree}

\Ccgs allow us to model the problem of planning data movement as a \emph{graph problem}.
This approach is very flexible: If there is \emph{any} way to connect execution operators via a sequence of conversion operators, we will discover it.
Unlike other approaches~\cite[\eg][]{gog2015musketeer,ires-bigdata},
developers do not need to provide conversion operators for all possible source and target channels.
It is therefore much easier for developers to add new platforms to \rheem and make them interoperable with the other platforms.
Let us further motivate the utility of \ccgs for data movement with a concrete example.

\begin{example}
Figure~\ref{figure:channel-conversion-graph} shows an excerpt of \rheem's default CCG that is used to determine how to move data from a \textsf{JavaMap} execution operator ({\em root}) to a \add{\textsf{FlinkReduce}} ({\em target$_1$}) and a \add{\textsf{SparkMap}} execution operator ({\em target$_2$}).
While the {\em root} produces a Java \at{Stream} as output channel, $target_1$ and $target_2$ accept only a \add{Flink \at{DataSet}} and a (cached) \at{RDD}, respectively, as input channels.
Multiple conversions are needed to serve the two \add{target} operators.
\label{example:ccg}
\end{example}
\vspace{-0.2cm}

\begin{figure}[t]
	\centering
	\includegraphics[width=\linewidth]{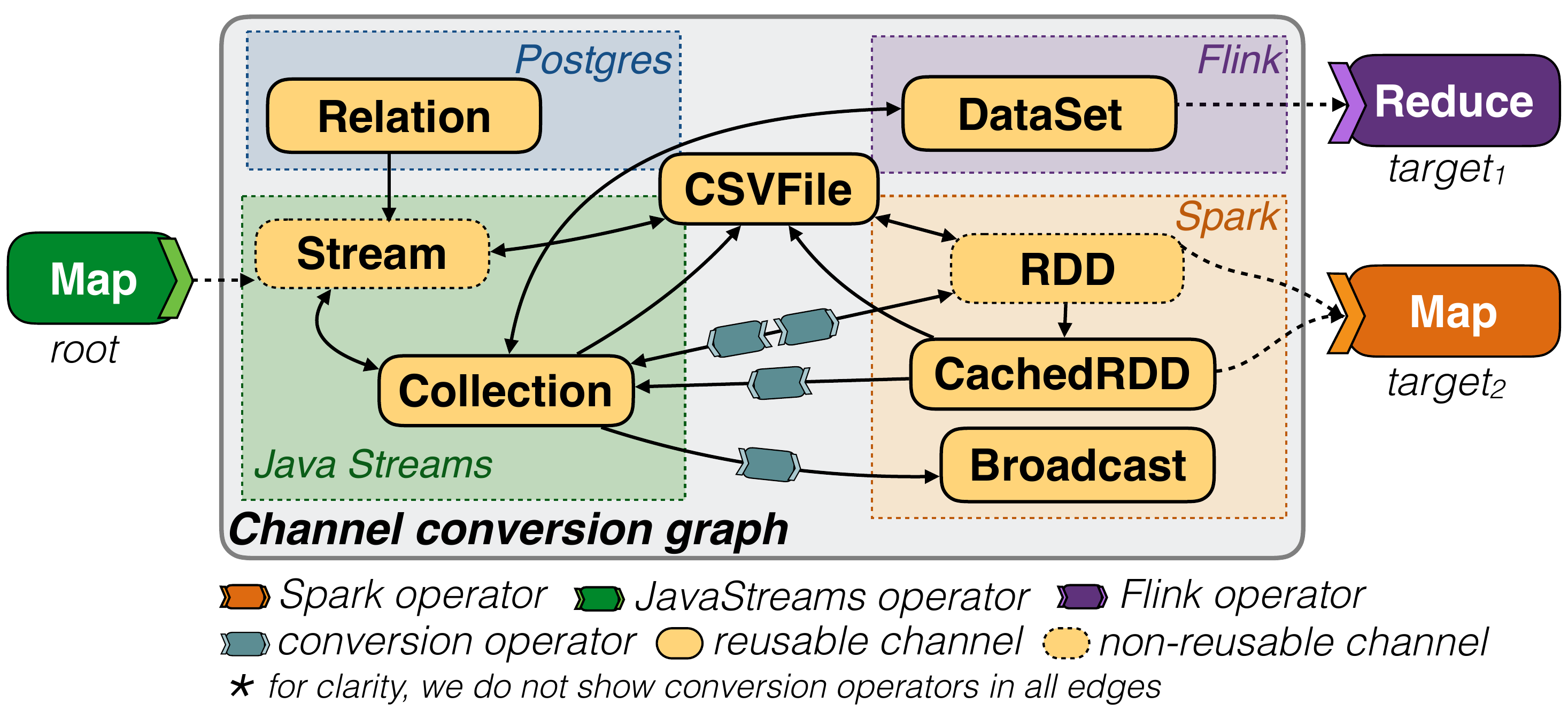}
	\vspace{-0.6cm}
	\caption{Example channel conversion graph along with root and target operators from different platforms.
	\label{figure:channel-conversion-graph}}
	\vspace{-0.2cm}
\end{figure}

Note that the \ccg also enables the optimizer to use multiple intermediate steps to connect two operators. For example, for transferring data from \pl{Postgres} to \pl{Flink} or \pl{Spark} in Figure~\ref{figure:channel-conversion-graph}, there are two intermediate channels involved, \ie~\at{Stream} and \at{Collection}.
We model such complex scenarios of finding the most efficient communication path from a root producer to multiple target consumers as the \emph{minimum conversion tree} (MCT) problem.

\begin{problem1}
Given a \emph{root channel} $c_r$, $n$~\emph{target channel sets} $C_{t_i}$ ($0 < i \leq n$), and the \ccg $G = (C, E, \lambda)$, find a subgraph $G'$ (\ie~a minimum conversion tree), such that:
\begin{packed_enum}
\item $G'$ is a directed tree with root $c_r$ and contains at least one channel $c_{t_i}$ for each target channel set $C_{t_i}$, where $c_{t_i}\in C_{t_i}$.
\item Any non-reusable channel in $G'$, must have a single successor, \ie~a conversion or a consumer operator. \label{prop:reuse} 
\item The sum of costs of all edges in $G'$ is minimized, \ie~there is no other subgraph $G''$ that satisfies the above two conditions and has a smaller cost than $G'$.
The cost of an edge $e$ is the estimated cost for the associated conversion operator $\lambda(e)$. 
\end{packed_enum}
\label{def:mtc}
\end{problem1}

\begin{example}\label{example:mct}
	In the example of Figure~\ref{figure:channel-conversion-graph}, the root channel is $c_r := \at{Stream}$ and the target channel sets are \add{$C_{t_1} := \{\at{DataSet}\}$} (for target$_1$) and $C_{t_2} := \{\at{RDD}, \at{CachedRDD}\}$ (for target$_2$).
	A minimum conversion tree for this scenario could look as follows:
	The \at{Stream} root channel is converted to a Java \at{Collection}.
	This \at{Collection} is then converted twice; namely to a \add{Flink \at{DataSet}} (thereby satisfying $C_{t_1}$) and to an \at{RDD} (thereby satisfying $C_{t_2}$).
	Note that this is possible only because \at{Collection} is reusable.
\end{example}

Although our MCT problem seems related to other well-studied graph problems, such as the minimum spanning tree and single-source multiple-destinations shortest paths, it differs from them for two main reasons.
First, MCTs have a fixed root and need not span the whole \ccg.
Second, MCT seeks to minimize the costs of the conversion tree as a whole rather than its individual paths from the root to the target channels.
It is the Group Steiner Tree (GST) problem\add{~\cite{group-steiner}} that is the closest to our MCT problem: There, $n$ sets of vertices should be connected by a minimal tree.
However, this problem is typically considered on undirected graphs and without the notion of non-reusable channels.
Furthermore, GST solvers are often designed only for specific types of graphs, such as planar graphs or trees.
These disparities preclude the adaption of existing GST solvers to the MCT problem.
However, the GST problem allows to show the NP-hardness of our MCT problem.

\begin{restatable}{theorem}{graphproblem}
\label{theorem:graphproblem}
The MCT problem is NP-hard.
\end{restatable}
\begin{proof}
	The NP-hard problem of GST~\cite{group-steiner} can be reduced in polynomial time to an MCT problem.
	Recall a GST instance consists of a weighted graph $G$ with positive edge weights, a root vertex $r$, and $k$ subsets (groups) of vertices from $G$.
	The goal of GST is to find a tree $G'$ on $G$ that connects $r$ with at least one vertex of each group.
	We convert an instance of GST to MCT as follows.
	We provide as input to MCT
	(i)~a channel conversion graph that has exactly the same vertices and edges with $G$,
	(ii)~the vertex $r$ as root channel,
	(iii)~the $k$ groups as target channel sets, and
	(iv)~the edge weights of the graph as conversion costs.
	This conversion is clearly of polynomial complexity.
\end{proof}

\subsection{Finding Minimum Conversion Trees}
\label{section:algorithm}
Because the MCT problem differs from existing graph problems, we devise a new algorithm to solve it (Algorithm~\ref{algorithm:shortest-tree}).
Given a \ccg $G$, a root channel $c_r$, and $n$ target channel sets $\mathscr{C}_t := \{C_{t_1}, C_{t_2}, ..., C_{t_n}\}$, the algorithm proceeds in two principal steps.
First, it simplifies the problem by modifying the input parameters ({\em kernelization}, Line~1).
Then, it \add{exhaustively} explores the graph ({\em channel conversion graph exploration}, Line~2) to find the MCT (Line~3).
We discuss these two steps in the following. 

\begin{algorithm}[t!]\small
	\SetKwComment{Comment}{$\triangleright$\ }{}
	\caption{Minimum conversion tree search.}
	\label{algorithm:shortest-tree}
	\KwIn{conversion graph $G$, root channel $c_r$, target channel sets $\mathscr{C}_t$}
	\KwOut{minimum conversion tree}
	\SetKwProg{Fn}{Function}{ }{ }
	$\mathscr{C}_t\leftarrow\texttt{kernelize}(\mathscr{C}_t)$\;
	$T_{c_r} \leftarrow \texttt{traverse}(G, c_r, \mathscr{C}_t, \emptyset, \emptyset)$\;
	\Return{$T_{c_r}[\mathscr{C}_t]$}\;
	\rememberlines
\end{algorithm}


\myparagraph{Kernelization}
In the frequent case that two (or more) target consumers target$_i$ and target$_j$ accept the same channels, \ie~$C_{t_i}=C_{t_j}$, with at most one non-reusable channel and at least one reusable channel, we can merge them into a single set by discarding the non-reusable channel: $C_{t_{i,j}}=\left\{c~\mid c~\in~C_{t_i} \wedge c \text{ is reusable}\right\}$.
The key point of this kernelization is that it decreases the number of target channel sets and thus, reduces the maximum degree (fanout) of the MCT, which is a major complexity driver of the MCT problem.
In fact, in the case of only a single target channel set the MCT problem becomes a single-source single-destination shortest path problem, which we can solve with, \eg~Dijkstra's algorithm.

\begin{example}[Merging target channel sets] In Figure~\ref{figure:channel-conversion-graph}, $\textsf{target}_2$ accepts the channels $C_{t_2}=\{\textsf{RDD}, \textsf{CachedRDD}\}$.
\add{Assume that the other consumer, $\textsf{target}_1$ would be a \emph{SparkReduce} operator instead, which accepts the same set of channels as $\textsf{target}_2$.}
In this case, we can merge their input channels into \add{$C_{t_{1,2}}=\{\textsf{CachedRDD}\}$}.
\end{example}

\begin{lemma}\label{lemma:kernelization}
A solution for a kernelized MCT problem also solves the original MCT problem.
\begin{proof}
	Assume an original MCT problem $M_o$ with target channel sets $C_{t_1}$, \dots, $C_{t_k}$ and a kernelized MCT problem $M_k$ for which those $C_{t_i}$ have been merged to a single target channel set $C^{t*}$.
	Now let $t_k$ be an MCT for $M_k$.
	Obviously, $t_k$ is also a conversion tree for $M_o$, but it remains to show that it is also minimum.
	For that purpose, we assume that $t_k$ was not minimum for $M_o$; in consequence, there has to be some other MCT $t_o$ for $M_o$.
	If $t_o$ satisfies all target channel sets of $M_o$ (\ie~the $C_{t_i}$) via the same communication channel $c$, then $t_o$ would also be an MCT for $M_k$, which contradicts our assumption.
	Specifically, $c$ must be a reusable channel, as it satisfies multiple target channel sets.
	In contrast, if $t_o$ satisfies the target channel sets of $M_o$ with different channels, then there has to be at least one reusable channel $c'$ among them, because we kernelize only such target channel sets that have \emph{at most} one non-reusable channel.
	Because $c'$ alone can already satisfy all target channel sets of $M_o$, it follows that $t_o$ produces more target channels than necessary and is therefore not minimal -- which also contradicts our assumption.
\end{proof}
\end{lemma}


\myparagraph{Channel conversion graph exploration}
After kernelizing the original MCT problem, Algorithm~\ref{algorithm:shortest-tree} proceeds to explore the CCG, thereby building the MCT from ``its leaves to the root'':
Intuitively, our algorithm searches -- starting from the root channel $c_r$ -- across the CCG for communication channels that satisfy the target channel sets $\mathscr{C}_t$;
It then backtracks the search paths, thereby incrementally building up the MCT.
The \texttt{traverse} function implements this strategy via recursion -- in other words, each call of this function represents a recursive traversal step through the \ccg.
In summary, the \texttt{traverse} function is composed of three main parts:
(i)~it visits a new channel, checks if it belongs to any target channel set, and potentially creates a partial singleton conversion tree;
(ii)~then it traverses forward, thereby creating {\em partial} MCTs from the currently visited channel to any subset of target channel sets; and
(iii)~it merges the partial MCTs from the steps (i) and (ii) and returns the {\em merged} MCTs.
The algorithm terminates when the partial MCTs form the final MCT.


\begin{algorithm}[t!]\small
	\SetKwComment{Comment}{$\triangleright$\ }{}
	\caption{Recursive traversal of MCT of Algorithm~\ref{algorithm:shortest-tree}.}
	\label{algorithm:traverse}
	\resumenumbering
	\SetKwProg{Fn}{Function}{ }{ }
\KwIn{channel conversion graph $G$, current channel $c$, target channel sets $\mathscr{C}_t$, visited channels $C_v$, satisfied target channel sets $\mathscr{C}_s$}
\KwOut{minimum conversion trees from $c$ to subsets of $\mathscr{C}_t$}
\Fn{\texttt{traverse}($G, c, \mathscr{C}_t, C_v, \mathscr{C}_s$)}{
	$T \leftarrow\texttt{create-dictionary}()$\;
	$\mathscr{C}'_s \leftarrow \{C_{t_i} \in \mathscr{C}_t \mid c \in C_{t_i}\} \setminus \mathscr{C}_s$\;
	\If{$\mathscr{C}'_s \neq \emptyset$}{
		\lForEach{$\mathscr{C}''_s \in 2^{\mathscr{C}'_s} \setminus \emptyset$}{
			$T[\mathscr{C}''_s] \leftarrow \texttt{tree}(c)$
		}
		\lIf{$\mathscr{C}_s \cup \mathscr{C}'_s = \mathscr{C}_t$}{
			\Return{$T$}
		}
	}
	$C_v\leftarrow C_v \cup \{c\}$ \;
	\lIf{$\texttt{reusable}(c)$}{$\mathscr{C}_s \leftarrow\mathscr{C}_s \cup \mathscr{C}'_s$} 
	$\mathcal{T} \leftarrow \emptyset$\;
	\ForEach{$(c\stackrel{o}{\rightarrow}c') \in G$ with $c' \not\in C_v$}{
		$T' \leftarrow \texttt{traverse}(G, c', \mathscr{C}_t, C_v, \mathscr{C}_s)$\;
		$T' \leftarrow \texttt{grow}(T', c\stackrel{o}{\rightarrow}c')$\;
		$\mathcal{T} \leftarrow \mathcal{T} \cup \{T'\}$\;
	}
	\leIf{$\texttt{reusable}(c)$}{
		$\mathit{d} \leftarrow |\mathscr{C}_t| - |\mathscr{C}_s|$
	}{$\mathit{d} \leftarrow 1$}
	\ForEach{$\mathbf{T} \in \texttt{disjoint-combinations}(\mathcal{T}, d)$}{
		$T\leftarrow\texttt{merge-and-update}(\mathbf{T}, T)$
	}
	\Return{$T$}\;
}
\end{algorithm}

We now explain in further detail this \texttt{traverse} function. 
The objective of each recursion step is to build up a dictionary $T$ (Line~5) that associates subsets of the target channel sets, \ie~$\mathscr{C}_s \subseteq \mathscr{C}_t$, with \emph{partial} conversion trees (PCTs) from the currently visited channel to those target channels $\mathscr{C}_s$.
While backtracking from the recursion, these PCTs can then be merged successively until they form the final MCT.
We use the following example to further explain Algorithm~\ref{algorithm:traverse}. 

\vspace{-0.2cm}
\begin{example}\label{example:traverse}
	Assume we are solving the MCT problem in Figure~\ref{figure:channel-conversion-graph},
	\ie~$c_r := \textsf{Stream}$,
	\add{$C_{t_1} := \{\at{DataSet}\}$}, and
	$C_{t_2} := \{ \at{RDD}, \at{CachedRDD} \}$.
	Also, assume that we have already made one recursion step from the \textsf{Stream} to the \textsf{Collection} channel.
	That is, in our current invocation of \texttt{traverse} we visit $c := \textsf{Collection}$, on our current path we have visited only $C_v = \{ \textsf{Stream} \}$ and did not reach any target channel sets, \ie~$\mathscr{C}_s := \emptyset$.
\end{example}
\vspace{-0.2cm}

\vspace{1mm}\noindent
\textit{Visit channel (Lines~6--9).}
The \texttt{traverse} function starts by collecting all so far unsatisfied target channel sets $\mathscr{C}'_s$, that are satisfied by the currently visited channel $c$ (Line~6).
If there is any such target channel set (Line~7), we create a PCT for any combinations of those target channel sets in $\mathscr{C}'_s$ (Line~8).
At this point, these PCTs consist only of $c$ as root node, but will be ``grown'' during backtracking from the recursion.
If we have even satisfied \emph{all} target channel sets on our current traversal path, we can immediately start backtracking (Line~9).
For the Example~\ref{example:traverse}, $c = \textsf{Stream}$ does not satisfy any target channel set, \ie~we get $\mathscr{C}'_s = \emptyset$ and need to continue.

\vspace{1mm}\noindent
\textit{Forward traversal (Lines~10--16).}
In the second phase, the \texttt{traverse} function does the \emph{forward} traversal.
For that purpose, it marks the currently visited channel $c$ as visited;
and if $c$ is reusable \emph{and} satisfies some target channel sets $\mathscr{C}'_s$, it marks those sets also as satisfied (Lines~10--11).
This is important to let the recursion eventually terminate.
Next, the algorithm traverses forward by following all \ccg edges starting at $c$ and leading to an unvisited channel (Lines~13--14).
For the Example~\ref{example:traverse}, we accordingly visit \add{\at{DataSet}}, \at{Broadcast}, \at{RDD}, and \add{\at{CSVFile}}.
Each recursive call yields another dictionary $T'$ of PCTs.
For instance, when invoking $\texttt{traverse}$ on \add{\at{DataSet}, we get $T'[C_{t_1}] = \at{DataSet}$ (a PCT consisting only of \at{DataSet} as root).}
At this point, we add the followed edge to this PCT to ``grow'' it (Line~16) and obtain the PCT \add{$\at{Collection} \rightarrow \at{DataSet}$.}
We store all those ``grown'' PCTs in $\mathcal{T}$.

\vspace{1mm}\noindent
\textit{Merge PCTs (Lines~17--20).}
As a matter of fact, none of the PCTs in $\mathcal{T}$ might have reached all target channel sets.
For instance, the above mentioned PCT \add{$\at{Collection} \rightarrow \at{DataSet}$} is the only one to satisfy $C_{t_1}$, but it does not satisfy $C_{t_2}$.
Thus, the third and final phase of the \texttt{traverse} function merges certain PCTs in $\mathcal{T}$.
%
Specifically, the $\texttt{disjoint-combinations}$ function (Line~18) enumerates all combinations of PCTs in $\mathcal{T}$ that
(i)~originate from different recursive calls of \texttt{traverse};
(ii)~do not overlap in their satisfied target channel sets;
and (iii)~consist of 1 to $d$ different PCTs.
While the former two criteria ensure that we enumerate all combinations of PCTs that may be merged, the third criterion helps us to avoid enumerating \emph{futile} combinations:
When the current channel $c$ is not reusable, it must not have multiple consuming conversion operators, so $d$ is set to 1 (Line~17).
In any other case, any PCT must not have a degree larger than the number of not satisfied target channels sets;
otherwise the enumerated PCTs would overlap in their satisfied target channel sets.
Note that the value of $d$ can be lowered by kernelization, which reduces the number of target channel sets.
For the Example~\ref{example:traverse}, we have \add{four} outgoing conversion edges from $c = \at{Collection}$ but only two non-satisfied target channel sets, namely $C_{t_1}$ and $C_{t_2}$.
As a result, we can avoid merging PCTs from all \add{four} edges \emph{simultaneously}, as the resulting PCT could not be minimal.
Eventually, the \texttt{merge-and-update} function combines the PCTs into a new PCT and, if there is no PCT in $T$ already that reaches the same target channel sets and has lower costs, the new PCT is added to $T$ (Line~19).
Amongst others, we merge the PCTs \add{$\at{Collection} \rightarrow \at{DataSet}$} and $\at{Collection} \rightarrow \at{RDD}$ in our example.
When we backtrack (Line~20), the resulting PCT will be ``grown'' by the edge $\at{Stream} \rightarrow \at{Collection}$ and form the eventual MCT.

\begin{restatable}{theorem}{conversionthm} \label{theorem:conversion}
	Given a channel conversion graph, Algorithm~\ref{algorithm:shortest-tree} finds the minimum conversion tree if it exists.
\end{restatable}

\begin{proof}
	As per Lemma~\ref{lemma:kernelization}, the kernelization does not change the solution of an MCT problem, so we proceed to prove the correctness of the graph traversal algorithm -- by induction.
	Let $h$ be the height of the MCT.
	If $h=1$, the conversion tree, which is composed of only a root (cf.~Algorithm~\ref{algorithm:shortest-tree}, Line~8), is always minimal as any conversion operator incurs non-negative costs.
	Assume an MCT of height $h$.
	We prove that our algorithm can output a tree of height $h+1$ that is also minimal.
	When merging PCTs two facts hold:
	(i)~any subtree in the MCT must be an MCT (with its own root), otherwise this subtree has a cheaper alternative and the overall conversion tree cannot be minimal; and
	(ii)~we consider all valid combination of PCTs in the merging phase and hence will not miss out the most efficient combination.
	Thus, given an MCT with height $h$, the tree with height $h+1$ will also be minimal.
\end{proof}

\myparagraph{Complexity and correctness}
Our algorithm solves the MCT problem exactly (see Theorem~\ref{theorem:conversion} below).
This comes at the cost of exponential complexity:
There are $(n - 1)!$ ways to traverse a full \ccg of $n$ channels and we might need to maintain $2^k$ partial trees in the intermediate steps, where $k$ is the number of target channel sets.
However, in practical situations, our algorithm finishes in the order of milliseconds, as the \ccg comprises only tens of channels and is very sparse.
Also, the number of target channel sets $k$ is mostly only 1 or 2 and can often be diminished by the kernelization.
More importantly, our algorithm avoids performance penalties from inferior data movement plans.
However, if it ever runs into performance problems, one may consider making it approximate.
Inspiration could be drawn from existing algorithms for GST~\cite{polylogarithmic-group-steiner,greedy-group-steiner}.
Yet, we evaluate our algorithm's scalability in Section~\ref{section:experiments_indepth} and show that it gracefully scales to a reasonable number of platforms.

\section{Plan Enumeration}
\label{section:enumeration}

The goal of our optimizer is to find the optimal plan, \ie~the plan with the smallest estimated cost.
More precisely, for each inflated operator in an inflated plan, it needs to select one of its alternative execution operators, such that the overall execution cost is minimized.
Finding the optimal plan, however, is challenging because of the exponential size of the search space.
A plan with $n$ operators, each having $k$ execution operators, will lead to $k^n$ possible execution plans.
This number quickly becomes intractable for growing $n$.
For instance, a cross-community Page\-Rank plan, which consists of $n{=}27$ operators, each with $k{=}5$, yields $2,149,056,512$ possible execution plans.
One could apply greedy pruning to reduce the search space significantly.
For example, we could pick only the most cost-efficient execution operators for each inflated operator and prune all plans with other execution operators, but such a greedy approach could not guarantee to find the optimal execution plan, because it neglects data movement and platform start-up costs.

Thus, it is worthwhile to spend a bit more computation time in the optimization process in order to gain significant performance improvements in the task execution.
We take a principled approach to solve this problem:
We define an algebra to formalize the enumeration (Section~\ref{sec:enumeration_algebra}) and propose a lossless pruning technique (Section~\ref{sec:enumeration_pruning}).
We then exploit this algebra and pruning technique to devise an efficient enumeration algorithm (Section~\ref{sec:enumeration_algorithm}).
\add{Intuitively, the plan enumeration process builds execution plans incrementally from an inflated \rheem plan.
It starts with subplans consisting of a single inflated operator and unfolds it with all possible execution operators.
Then, it expands the subplans with their neighboring operators using the algebra until all operators have been unfolded.}

\subsection{Plan Enumeration Algebra}
\label{sec:enumeration_algebra}
Inspired by the relational algebra, we define the plan enumeration search space along with traversal operations algebraically.
This approach enables us to:
(i)~define the enumeration problem in a simple, elegant manner;
(ii)~concisely formalize our enumeration algorithm; and
(iii)~explore design alternatives.
Let us first describe the data structures and operations of our algebra.

\myparagraph{Data structures}
Our enumeration algebra needs only one principal data structure, the \emph{enumeration} $E=(S, SP)$, which comprises a set of \emph{execution subplans} $SP$ for a given \emph{scope} $S$.
\add{The scope is the set of inflated operators that the enumeration has unfolded in the current step, while} 
each subplan \add{contains} execution operators for each inflated operator in $S$, including execution operators for the data movement.
Intuitively, an enumeration can be seen as a relational table whose schema corresponds to its scope and whose tuples correspond to its possible execution subplans.

\begin{figure}[h!]
  \vspace{-0.2cm}
  \includegraphics[scale=0.13]{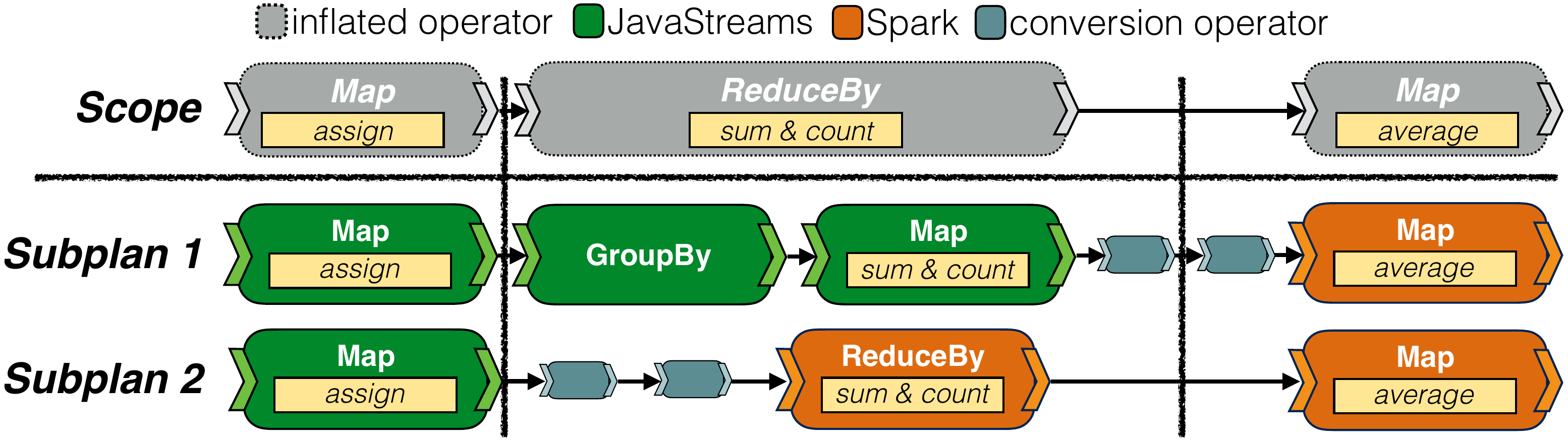}
  \vspace{-0.4cm}
  \caption{An example enumeration with two subplans.}
  \label{figure:enumeration}
  \vspace{-0.2cm}
\end{figure}

\begin{example}[Enumeration]
Figure~\ref{figure:enumeration} depicts an enumeration for the operators from Figure~\ref{figure:enrichment}.
It comprises two subplans for a scope of three inflated operators.
\end{example}
\vspace{-0.2cm}

Notice that if the scope contains all the inflated operators of a \rheem\ plan \add{(\em{complete enumeration})}, then the corresponding subplans form complete execution plans.
This admits the following problem formalization.

\begin{problem2}
Given a \rheem plan, let $E = (S, \textit{SP})$ be its complete enumeration.
The goal is to efficiently identify a subplan $\textit{sp}_k \in \textit{SP}$ such that $\textit{cost}(\textit{sp}_k) \leq \textit{cost}(\textit{sp}_i)$  $\forall \textit{sp}_i \in \textit{SP}$,
where $\textit{cost}(\textit{sp}_i)$ comprises the costs of execution, data movement, and platform initializations of $\textit{sp}_i$.
\end{problem2}

\myparagraph{Algebra operations}
\add{To be able to expand an enumeration with the neighboring operators of its subplans we use our enumeration algebra.}
It comprises two main operations, {\em Join} ($\bowtie$) and {\em Prune} ($\sigma$), both of which allow to manipulate enumerations.
In few words, {\em Join} connects two small enumerations to form a larger one, while {\em Prune} scraps inferior subplans from an enumeration for efficiency reasons.
Let us briefly establish these two operations before detailing how they can be used to enumerate complete execution plans.


\begin{definition}[Join]
Given two disjoint enumerations $E_1=(S_1, SP_1)$ and $E_2=(S_2, SP_2)$ (\ie~$S_1 \cap S_2 = \emptyset$), we define a join $E_1 \bowtie E_2 = (S, SP)$ where $S := S_1 \cup S_2$ and
$SP := \{\texttt{connect}(sp_1, sp_2)\mid sp_1 \in SP_1 \text{ can be connected to } sp_2 \in SP_2 \}$.
The \texttt{connect} function connects $sp_1$ and $sp_2$ by adding conversion operators between operators of the two subplans as explained in Section~\ref{section:datamovement}.

	\label{def:join}
\end{definition}


%

\begin{example}[Merging subplans]
  The enumeration in Figure~\ref{figure:enumeration} could be created by joining an enumeration with scope $S_1 = \{\at{Map} \textit{(``assign'')}, \at{ReduceBy} \textit{(``sum{\&}count'')}\}$ with an enumeration with scope $S_2 = \{\at{Map} \textit{(``average'')}\}$.
  In particular, the \texttt{connect} function adds conversion operators to link the two \at{Maps} in Subplan~1.
\end{example}


\begin{definition}[Prune]
	Given an enumeration $E=(S, SP)$, we define a pruned enumeration $\sigma_\pi(E) := (S, SP')$, where $SP':= \{ sp \in SP \mid sp\text{ satisfies } \pi \}$ and $\pi$ is a configurable pruning criterion.
\end{definition}

\myparagraph{Applying the algebra}
Let us now draft a basic enumeration algorithm based on the \emph{Join} and \emph{Prune} operations.
For each inflated operator~$o$, we create a singleton enumeration $E~=~(\{o\},~SP_o)$, where $SP_o$ are the executable subplans provided by~$o$.
We then join these singleton enumerations one after another to obtain an exhaustive enumeration for the complete \rheem plan.
By pruning the enumerations before joining them, we can drastically reduce the number of intermediate subplans, which comes with according performance benefits.
That being said, this algorithm still lacks two important details, namely a concrete pruning rule and an order for the joins.
We present our choices on these two aspects in the remainder of this section.

\subsection{Lossless Pruning}
\label{sec:enumeration_pruning}
We devise a novel strategy for the \emph{Prune} operation that is {\em lossless}: it will not prune a subplan that is part of the optimal execution plan.
As a result, the optimizer can find the optimal execution plan without an exhaustive enumeration of all execution plans.
Overall, our pruning technique builds upon the notion of {\em \open} operators, which are inflated operators of an enumeration with scope $S$ that are \emph{adjacent} to some inflated operator \emph{outside} of $S$.

\begin{example}[Boundary operators]
  In the scope of the enumeration from Figure~\ref{figure:enumeration},
  \at{Map} (``assign'') and \textsf{Map} (``average'') are \open operators, because they have adjacent operators outside the scope; namely \textsf{RepeatLoop} and \textsf{Map} (``parse'') (cf.\@ Figure~\ref{figure:plan}).
\end{example}

Having explained the boundary operators, we proceed to define our lossless pruning strategy that builds upon them.

\begin{definition}[Lossless Pruning]\label{definition:lossless-pruning}
Let $E=(S, SP)$ be an enumeration and $S_b \subseteq S$ be the set of its \emph{\open} operators.
The lossless pruning removes all $sp\in SP$ for which there is another $sp' \in SP$ that
(i)~contains the same execution operators for all $S_b$ as $sp$,
(ii)~employs the same platforms as $sp$, and
(iii)~has lower cost than $sp$.
\end{definition}

\begin{example}[Lossless Pruning]
For our example enumeration from Figure~\ref{figure:enumeration}, the lossless pruning discards either Subplan~1 or Subplan~2 (whichever has the \add{higher} cost), because
(i)~the two subplans contain the same boundary execution operators (\at{JavaMap} (``assign'') and \at{SparkMap} (``average'')); and (ii)~they need to initialize the same platforms (Java Streams and Spark).
\label{example:pruning}
\end{example}

This pruning technique effectively renders the enumeration a dynamic programming algorithm by establishing the principle of optimality for certain subplans.
Let us now demonstrate that this pruning rule is indeed lossless.

\begin{restatable}{lemma}{pruninglem}\label{lemma:pruning}
	The lossless pruning does not prune a subplan that is contained in the optimal plan {w.r.t.} the cost model.
\end{restatable}
\begin{proof}
	We prove this lemma by contradiction.
	Consider an enumeration $E = (S, \textit{SP})$ and two execution subplans $\textit{sp}_1, \textit{sp}_2\in\textit{SP}$.
	Let us assume that both subplans share the same \open operators and use the same platforms but $\textit{sp}'$ has a lower cost than $\textit{sp}$, so that our pruning removes $\textit{sp}$.
	Now assume that the subplan $\textit{sp}$ is contained in the optimal plan $p$.
	If we exchange $\textit{sp}$ with $\textit{sp}'$, we obtain a new plan $p'$.
	This plan is valid because $\textit{sp}$ and $\textit{sp}'$ have the same \open operators, so that any data movement operations between $\textit{sp}$ with any adjacent operators in $p$ are also valid for $\textit{sp}'$.
	Furthermore, $p'$ is more efficient than $p$ because the costs for $\textit{sp}'$ are lower than for $\textit{sp}$ and besides those subplans, $p$ and $p'$ have the exact same operators and costs.
	This contradicts the assumption that $p$ is optimal.
\end{proof}

\subsection{Enumeration algorithm}
\label{sec:enumeration_algorithm}
\add{Using the previously described enumeration algebra and the lossless pruning strategy we construct our enumeration algorithm.}
Algorithm~\ref{algorithm:enumeration} shows the algorithm.
Given an inflated \rheem plan, we first create a singleton enumeration for each inflated operator (Line~1).
We then need to repeatedly join and prune these enumerations to obtain the optimal execution plan.
However, we aim at maximizing the pruning effectiveness by choosing a good order to join the enumerations.
Thus, we first identify \emph{join groups}~(Line~2).
A join group indicates a set of plan enumerations to be joined.
Initially, we create a join group for each inflated operator's output, so that each join group contains
(i)~the enumeration for the operator with that output, $E_\text{out}$, and
(ii)~the enumerations for all inflated operators that consume that output as input, $E_\text{in}^i$.
For instance in the inflated plan of Figure~\ref{figure:plan}, the enumerations for \at{Map} (``assign'') and \at{ReduceBy} (``sum \& count'') form an initial join group.
While the join order is not relevant to the correctness of the enumeration algorithm, joining only adjacent enumerations is beneficial to performance: It minimizes the number of \open operators in the resulting enumeration, which in turn makes our lossless pruning most effective (see Definition~\ref{definition:lossless-pruning}, Criterion~(i)).
To further promote this effect, we order the join groups ascending by the number of \open operators~(Line~3).
Then, we greedily poll the join groups from the queue, execute the corresponding join, and prune the join product~(Lines~4--6).
Also, in any other join group that includes one of the joined enumerations, \ie~$E_\text{out}$ or any $E_\text{in}^i$, we need to replace those joined enumerations with the join product $E_{\bowtie}$~(Lines~7--9).
Note that these changes make it necessary to re-order the affected join products in the priority queue~(Line~10).
Eventually, the last join product is a full enumeration for the complete \rheem plan.
Its lowest cost subplan is the optimal execution plan (Line~11).

\begin{algorithm}[t!]\small
	\caption{\rheem plan enumeration}
	\label{algorithm:enumeration}
	\KwIn{\rheem\ inflated plan $R$}
	\KwOut{Optimal execution plan $sp_\text{min}$}

  $\mathcal{E} \leftarrow\big\{(\{o\}, \textit{SP}_o): o \text{ is an inflated operator} \in R \big\}$ \;
  $\textit{joinGroups} \leftarrow \texttt{find-join-groups}(\mathcal{E})$ \;
  $\textit{queue} \leftarrow \texttt{create-priority-queue}(\textit{joinGroups})$ \;
  \While{$|\textit{queue}| > 0$}{
    $\textit{joinGroup} = \{E_\text{out}, E_\text{in}^1, E_\text{in}^2, \dots \} \leftarrow \texttt{poll}(\textit{queue})$ \;
    $E_{\bowtie} \leftarrow \sigma(E_\text{out} \bowtie E_\text{in}^1 \bowtie  E_\text{in}^2 \bowtie \dots)$ \;
    \ForEach{$\textit{joinGroup}' \in \textit{queue}$}{
      \If{$\textit{joinGroup} \cap \textit{joinGroup}' \neq \emptyset$}{
        \texttt{update}($\textit{joinGroup}'$\textsf{ with }$E_{\bowtie})$ \;
        \texttt{re-order}($\textit{joinGroup}$\textsf{ in }$\textit{queue})$\;
      }
    }
  }
  $\textit{sp}_\text{min} \leftarrow \text{the subplan in }E_{\bowtie}\text{ with the lowest cost}$ \;
\end{algorithm}

\add{It is worth noting that our algorithm has been inspired by classical database optimizers~\cite{selinger1979access} with the difference that the problem we are solving is not operator re-ordering but rather choosing execution operators in a plan. For this reason, we do not opt for a top-down or bottom-up approach but rather exploit the entire search space simultaneously. 
}
\add{In addition,}  our lossless pruning is related to the concept of \emph{interesting sites}~\cite{kossmann2000iterative} in distributed relational query optimization, especially to the \emph{interesting properties}~\cite{selinger1979access} in general.
We can easily extend our pruning rule to account for properties other than \open operators.
For example, we already do consider platform start-up costs in our cost model
(see the plan enumeration problem statement in Section~\ref{sec:enumeration_algebra}).
As a result, we avoid pruning subplans with start-up costs that might be redeemed over the whole plan.
Let us now establish the correctness of our enumeration algorithm.


\begin{theorem}\label{theorem:enumeration}
The enumeration Algorithm~\ref{algorithm:enumeration} determines the optimal execution plan \wrt~the cost estimates.
\end{theorem}
\begin{proof}
As Algorithm~\ref{algorithm:enumeration} applies a lossless pruning technique (as per Lemma~\ref{lemma:pruning}) to an otherwise \emph{exhaustive} plan enumeration, it detects the optimal execution plan. 
\end{proof}

\section{Dealing with Uncertainty}
\label{section:progressive}


As cross-platform settings are characterized by high uncertainty, \eg~the semantics of UDFs are usually unknown, data cardinalities can be imprecise.
This harms the optimizer~\cite{howgoodoptimizersare}.
Although our optimizer allows users to supplement valuable optimization information, such as UDF selectivities, users might not always be willing or able to specify them.
Hence, the optimizer might choose suboptimal plans.


To mitigate the effects of bad cardinality estimates, our optimizer also performs \emph{progressive query optimization}~\cite{markl04}.
The key principle is to monitor actual cardinalities of an execution plan and re-optimize the plan on the fly in case of poor cardinality estimates.
Progressive query optimization in cross-platform settings is challenging for two reasons.
First, we have only limited control over the underlying platforms, which makes plan instrumentation and halting executions difficult.
Second, re-optimizing an ongoing execution plan must efficiently consider the results already produced.

Our optimizer tackles the above challenges as follows.
It first inserts \emph{optimization checkpoints} into execution plans.
An optimization checkpoint is basically a request for re-optimization before proceeding beyond it.
The optimizer inserts these checkpoints between two execution operators whenever (i)~cardinality estimates are uncertain (\ie~having a wide interval or low confidence) and (ii)~the data is at rest (\eg~a Java collection or a file).
Before execution, the optimizer asks the drivers of the involved platforms to collect the actual cardinalities of their intermediate data structures.
The execution plan is then executed until the optimization checkpoints.
Every time an optimization checkpoint is reached, the optimizer checks if the actual cardinalities considerably mismatch the estimated ones.
If so, it re-optimizes (as explained in previous sections) the plan under consideration with the updated cardinalities and already executed operators.
Once this is done, the \rheem executor simply resumes the execution with the re-optimized plan.
This yields a progressive optimization that always uses the latest statistics.




\section{Experiments}
\label{section:experiments}

Our optimizer is part of \rheem, our open-source cross-platform system\footnote{\small{\url{https://github.com/rheem-ecosystem/rheem}}}.
For the sake of simplicity, we henceforth refer to our optimizer simply as \rheem.
We have carried out several experiments to evaluate the effectiveness and efficiency of our optimizer.
As our work is the first to provide a complete cross-platform optimization framework, we compared it vis-a-vis individual platforms and common practices.
For a system-level comparison please refer to~\cite{rheem-system-vldb18}.

We evaluate our optimizer by answering the following questions.
Can our optimizer enable \rheem to:
{\em choose the best platform for a given task?} (Section~\ref{section:experiments_independence});
{\em spot hidden opportunities for cross-platform processing that improve performance?} as well as
{\em perform well in a data lake setting?} (Section~\ref{section:experiments_multiple}).
These are in fact the three most common situations in which an application needs support for cross-platform data processing~\cite{rheem-tutorial}.
Lastly, we also evaluate the scalability and design choices of our optimizer (Sections~\ref{section:experiments_indepth} and~\ref{section:extraexps}).


\subsection{Setup}
\label{section:experiments_setup}

\myparagraph{Hardware}
We ran all our experiments on a cluster of 10 machines.
Each node has one $2$\,GHz Quad Core Xeon processor, $32$\,GB main memory, $500$\,GB SATA hard disks, a $1$\,Gigabit network card and runs $64$-bit platform Linux Ubuntu 14\@.04\@.05.

\myparagraph{Processing \& storage platforms}
We considered the following platforms:
Java's Streams (\textsf{JavaStreams}),
PostgreSQL~9.6.2 (\textsf{PSQL}),
Spark~1.6.0 (\textsf{Spark}),
Flink~1.3.2 (\textsf{Flink}),
GraphX~1.6.0 (\textsf{GraphX}),
Giraph~1.2.0 (\textsf{Giraph}),
a simple self-written Java graph library (\textsf{JGraph}), and
HDFS~2.6.0 to store files.
We used all these with their default settings and set the RAM of each platform to $20$\,GB.



\begin{figure*}[!t]
	\centering
	\includegraphics[scale=0.3]{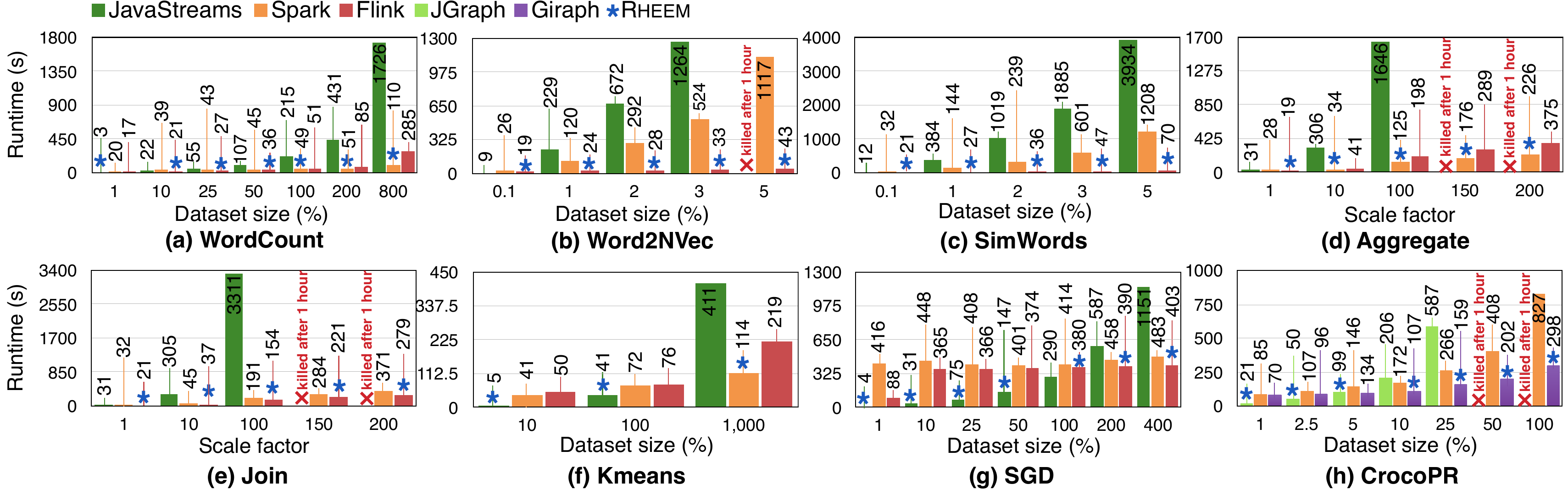}
	\vspace{-0.4cm}
	\caption{Platform independence: \rheem avoids all worst execution cases and chooses the best platform for almost all tasks.}
	\label{figure:platformindependence}
	\vspace{-0.3cm}
\end{figure*}

\myparagraph{Tasks and datasets}
We have considered a broad range of data analytics tasks from different areas, namely text mining (TM), relational analytics (RA), machine learning (ML), and graph mining (GM).
Details on the datasets and tasks are shown in Table~\ref{tab:datasets}.
These tasks and datasets individually highlight different features of \rheem and together demonstrate its general applicability.
Note that, to allow \rheem to choose most of the available platforms, all tasks' input datasets (with the exception of \task{Polystore}'s) are stored on HDFS (except when specified otherwise).
To challenge \rheem, we focused primarily on medium-sized datasets, so that \add{platform choices are not so obvious. Considering very large datasets would not yield very interesting insights: \eg~\pl{JavaStreams} or \pl{Postgres} could be easily excluded from the viable platform choices making the problem easier.}
Nonetheless, \rheem scales to large datasets when provided with scalable processing platforms.
\add{
	To learn the operators cost we first generated a number of execution logs using all tasks in Table~\ref{tab:datasets} with varying input dataset sizes and
	then used a genetic algorithm to learn the cost from these logs. 
}
Note that all the numbers we report are the average of three runs.

\begin{table}[ht]
	\centering
	\caption{Tasks and datasets.\label{tab:datasets}}
	\vspace{-0.3cm}
	\scalebox{0.7}{
		\begin{tabular}{lp{3cm}p{2.7cm}p{2cm}}
			\hline
			\textbf{Task} & \textbf{Description} & \textbf{Dataset} & \textbf{Default store} \\
			\hline
			\task{WordCount} (TM) & count distinct words & \ds{Wikipedia} \hspace{1cm} \ds{abstracts} ($3$GB) & \pl{HDFS}\\
			\task{Word2NVec} (TM) & word neighborhood vectors & \ds{Wikipedia} \hspace{1cm} \ds{abstracts} ($3$GB) & \pl{HDFS}\\
			\task{SimWords} (TM) & word neighborhood clustering & \ds{Wikipedia} \hspace{1cm} \ds{abstracts} ($3$GB) & \pl{HDFS}\\
		    \hline
 		    \task{Aggregate} (RA) & aggregate query (TPC-H Q1) & \ds{TPC-H} ($1$-$100$GB) & \pl{HDFS}\\
 		    \task{Join} (RA) & 2-way join (TPC-H Q3) & \ds{TPC-H} ($1$-$100$GB) & \pl{HDFS}\\
 		    \task{PolyJoin} (RA) & n-way join (TPC-H Q5) & \ds{TPC-H} ($1$-$100$GB) & \pl{Postgres}, \pl{HDFS}, \pl{LFS}\\
			\hline
			\task{Kmeans} (ML) & clustering & \ds{USCensus1990} ($361$MB) & \pl{HDFS}\\
			\task{SGD} (ML) & stochastic gradient descent & \ds{HIGGS} ($7.4$GB) & \pl{HDFS}\\
			\hline
			\task{CrocoPR} (GM) & cross-community pagerank & \ds{DBpedia pagelinks} ($20$GB) & \pl{HDFS}\\
			\hline
		\end{tabular}
	}
	\vspace{-0.2cm}
\end{table}


\begin{figure*}[!t]
	\centering
	\includegraphics[scale=0.295]{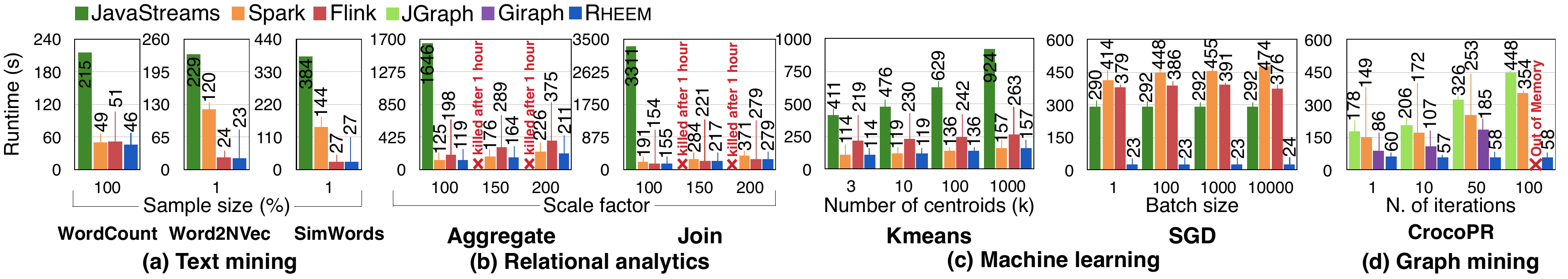}
	\vspace{-0.7cm}
	\caption{Opportunistic cross-platform: \rheem improves performance by combining multiple data processing platforms.}
	\label{figure:opportunistic}
	\vspace{-0.2cm}
\end{figure*}

\subsection{Single-Platform Optimization}
\label{section:experiments_independence}

We start our experiments by evaluating how well \rheem selects a single data processing platform to execute a given task.

\myparagraph{Experiment setup}
For this experiment, we forced \rheem to use a single platform when executing a task.
Then, we checked if our optimizer chose the one with the best runtime.
We ran all the tasks of Table~\ref{tab:datasets} with increasing dataset sizes. \add{For the real-world datasets, we take samples from the initial dataset of increasing size. To further stress the optimizer, for some tasks we increase the input datasets by replicating them.}
Note that we do not run \task{PolyJoin} as it cannot be performed using a single platform.
The iterations for \task{CrocoPR}, \task{K-means}, and \task{SGD} are $10$, $100$, and $1,000$, respectively.


\myparagraph{Results}
Figure~\ref{figure:platformindependence} shows the execution times for all our data analytic tasks and for increasing dataset sizes\footnote{For the non-synthetic datasets, we created samples of increasing size.}.
The stars denote the platform selected by our optimizer.
First of all, let us stress that the results show significant differences in the runtimes of the different platforms: even between \pl{Spark} and \pl{Flink}, which are big data platform competitors.
For example, \pl{Flink} can be up to $26$x faster than \pl{Spark} and \pl{Spark} can be twice faster than \pl{Flink} for the tasks we considered in our evaluation.
Therefore, it is crucial for an optimizer to prevent tasks from falling into such non-obvious worst cases.
The results, in Figure~\ref{figure:platformindependence}, show that our optimizer indeed makes robust platform choices whenever runtimes differ substantially.
This effectiveness of the optimizer for choosing the right platform transparently prevents applications from using suboptimal platforms.
For instance, it prevents running:
(i)~\task{Word2NVec} and \task{SimWords} on \textsf{Spark} for 5\% of its input dataset.
Spark performs worse than \textsf{Flink} for \task{Word2NVec} and \task{SimWords} because it employs only $2$ compute nodes (one for each input data partition), while \textsf{Flink} uses all $10$\add{\footnote{\add{One might think of re-partitioning the data for Spark, but such an optimization is the responsibility of the platform itself based on our three-layer optimization vision~\cite{rheem-vision-edbt16}.}}};
(ii)~\task{SimWords} on \textsf{Java} for 1\% of its input dataset ($\sim\!30$MB);
as \task{SimWords} performs many CPU-intensive vector operations, using \textsf{JavaStreams} (\ie~a single compute node) simply slows down the entire process;
(iii)~\task{WordCount} on \textsf{Flink} for 800\% of its input dataset (\ie~$24$GB), where, in contrast to \pl{Spark}, \pl{Flink} suffers from a slower data reduce mechanism\footnote{\pl{Flink} uses a sorting-based aggregation, which -- in this case -- appears to be inferior to \pl{Spark}'s hash-based aggregation.} ; and
(iv)~\task{CrocoPR} on \textsf{JGraph} for more than 10\% of its input dataset as it simply cannot efficiently process large datasets.
We also observe that \rheem generally chooses the right platform even for the difficult cases where the execution times are quite similar on different platforms.
For example, it always selects the right platform for \task{Aggregate} and \task{Join} even if the execution times for \textsf{Spark} and \textsf{Flink} are quite close to each other.
Only in few of these difficult cases the optimizer fails to choose the best platform, such as in \task{Word2NVec} and \task{SimWords} for $0.1$\% of input data.
This is because the accuracy of our optimizer is very sensitive to uncertainty factors, such as cost model calibration and cardinality estimates.
These factors are also quite challenging to estimate even for controlled settings, such as in databases.
Still, despite these two cases, all these results allow us to conclude that {\em our optimizer chooses the best platform for almost all tasks and it prevents tasks from falling into worst execution cases}.


\subsection{Multi-Platform Optimization}
\label{section:experiments_multiple}
We now study the efficiency of our optimizer when using multiple platforms for a single task.
We evaluate:
(i)~if it can spot hidden opportunities for the use of multiple platforms ({\em opportunistic} cross-platform); 
(ii)~when data is dispersed across several store platforms and no processing platform can directly access all the data ({\em polystore}); and
(iii)~the effectiveness to complement the functionalities of disparate processing platforms ({\em mandatory} cross-platform).


\myparagraph{Experiment setup (opportunistic)}
We re-enable \rheem to use any platform combination.
For the opportunistic cross-platform experiment, we use the same tasks and datasets with three differences: we ran
(i)~\task{Kmeans} on $10$x its entire dataset for a varying number of centroids,
(ii)~\task{SGD} on its entire dataset for increasing batch sizes, and
(iii)~\task{CrocoPR} on 10\% of its input dataset for a varying number of iterations.

\begin{table}[t!]
	\centering
	\caption{Opportunistic cross-platform breakdown.\label{tab:breakdown}}
	\vspace{-0.3cm}
	\scalebox{0.7}{
		\begin{tabular}{p{4.5cm} c c}
		\hline
		\textbf{Task} & \textbf{Selected Platforms} & \textbf{Data Transfer/Ite.} \\
		\hline
		\task{WordCount} & Spark, JavaStreams & $\sim\!82$ MB \\
		\task{Word2NVec} & Flink & -- \\
		\task{SimWords} & Flink & -- \\
		\hline
		\task{Aggregate (all scale factors)} & Flink, Spark & $\sim\!23$\% of the input \\
		\task{Join (all scale factors)} & Flink & -- \\
		\hline
		\task{Kmeans (k=3; k=10)} & Spark & -- \\
		\task{Kmeans (k=100, k=1,000)} & Spark, JavaStreams & $\sim\!6$ KB \& $\sim\!60$ KB resp. \\
		\task{SGD (all batch sizes)} & Spark, JavaStreams & $\sim\!0.14$ KB $\times$ batch size \\
		\hline
		\task{CrocoPR (all nb. of ite.)} & Flink, JGraph, JavaStreams & $\sim\!544$ MB \\
		\hline
		\end{tabular}
	}
	\vspace{-0.3cm}
\end{table}

\myparagraph{Results (opportunistic)}
Figure~\ref{figure:opportunistic} shows the results for these experiments.
Overall, we find that in the worst case \rheem matches the performance of any single platform execution,
but in several cases considerably improves over single-platform executions.
We observe it to be up to $20\times$ faster than \pl{Spark}, up to $17\times$ faster than \pl{Flink}, up to $21\times$ faster than \pl{JavaStreams}, up to $6\times$ faster than \pl{Giraph}.
There are several reasons for having this large improvement.
\add{Table~\ref{tab:breakdown} illustrates the platform choices as well as the cross-platform data transfer per iteration that our optimizer did for all our tasks.}

In detail, for \task{SGD}, \rheem decided to handle the model parameters, \add{which is typically tiny ($\sim\!\!0.1$KB for our input dataset)}, with \pl{JavaStreams} while it processed the data points (typically a large dataset) with \pl{Spark}.
For \task{CrocoPR}, surprisingly our optimizer uses a combination of \pl{Flink}, \pl{JGraph}, and \pl{JavaStreams}, even if \pl{Giraph} is the fastest baseline platform.
This is because after the preparation phase of this task, the input dataset for the \at{PageRank} operation on \pl{JGraph} is $\sim\!\!544$ MB only.
For \task{WordCount}, \rheem \add{ surprisingly} detected that moving the result data \add{($\sim\!82$ MB)} from \pl{Spark} to \pl{JavaStreams} and afterwards shipping it to the driver application is slightly faster than \pl{Spark} (which is the fastest baseline platform for this task).
This is because when moving data to \pl{JavaStreams} \rheem uses the action \textsf{Rdd.collect()}, which is more efficient than the \pl{Rdd.toLocalIterator()} operation that \pl{Spark} uses to move data to the driver.
For \task{Aggregate}, our optimizer selects \pl{Flink} and \pl{Spark}, which allows it to run this task slightly faster than the fastest baseline platform, which is \pl{Spark} for this task.
Our optimizer achieves this improvement by (i)~exploiting the fast stream data processing mechanism native in \pl{Flink} for the projection and selection operations, and (ii)~avoiding the slow data reduce mechanism of \pl{Flink} by using \pl{Spark} for the \at{ReduceBy} operation.
\add{Note that, in contrast to all previous tasks, \rheem can afford to transfer $\sim\!\!23$\% of the input data because it uses two big data platforms for processing this task.}
All these are surprising results per-se.
They show not only that \rheem outperforms state-of-the-art platforms by using combinations of them, but also that it can spot hidden opportunities for cross-platform execution.

To further stress the importance of finding hidden cross-platform execution opportunities, we ran a subquery (\task{JoinX}) of \task{PolyJoin}:
This query joins the relations {\tt SUPPLIER} and {\tt CUSTOMER} (which are stored on \pl{Postgres}) on the attribute {\tt nationkey} and aggregates the join results on the same attribute.
For this additional experiment, we compare \rheem with the execution of \task{JoinX} on \pl{Postgres}, which is the obvious platform to run this kind of queries.
\begin{wrapfigure}{lh}{0.13\textwidth}
	\vspace{-0.4cm}
	\centering
	\includegraphics[width=0.16\textwidth]{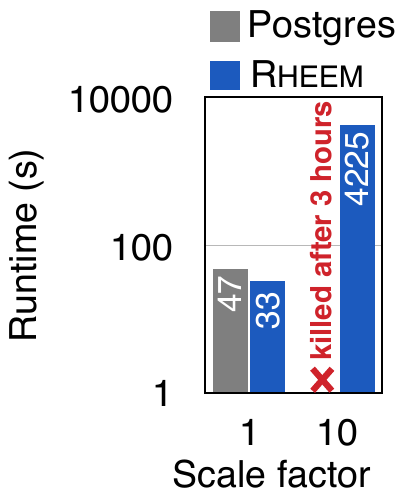}
	\vspace{-0.8cm}
	\caption{JoinX}
	\label{figure:joinx}
	\vspace{-0.4cm}
\end{wrapfigure}
\noindent
The results are displayed in Figure~\ref{figure:joinx}.
Remarkably, we observe that \rheem significantly outperforms \pl{Postgres}, even though the input data is stored there.
In fact, \rheem is more than twice as fast as \pl{Postgres} for a scale factor of $10$.
This is because it simply pushes down the projection operation into \pl{Postgres} and then moves the data into \pl{Spark} to perform the join and aggregation operations, thereby leveraging the parallelism offered by \pl{Spark}.
We thus do confirm that {\em our optimizer is indeed able to identify hidden opportunities to improve performance as well as to perform much more robustly by using multiple platforms}.

%



We now evaluate the efficiency of \rheem in polystore settings, where the input datasets are dispersed across several systems.

\myparagraph{Experiment setup (polystore)}
We consider the \task{PolyJoin} task, which takes the {\tt CUSTOMER}, {\tt LINEITEM}, {\tt NATION}, {\tt ORDERS}, {\tt REGION}, and {\tt SUPPLIER} TPC-H tables as input.
We stored the large \texttt{LINEITEM} and \texttt{ORDERS} tables in \pl{HDFS}, the \texttt{CUSTOMER}, \texttt{REGION}, and \texttt{SUPPLIER} tables in \pl{Postgres}, and the \texttt{NATION} table in a local file system (\pl{LFS}).
In this scenario, the common practice is either to move the data into a relational database in order to enact the analytical queries inside the database~\cite{polybase,oraclehadoop} or move the data entirely to HDFS and use Spark.
Therefore, we consider these two cases as the baselines.
We measure the data migration time as well as the query execution time as the total runtime for these baselines.
\rheem processes the input datasets directly on the data stores where they reside \add{and moves data if necessary}.
For a fair comparison in this experiment, we set the {\em parallel query} and {\em effective IO concurrency} features of \pl{Postgres} to $4$.


\myparagraph{Results (polystore)}
Figure~\ref{fig:polymandatory}(a) shows the results:
\rheem is significantly faster, namely up to $5\times$, \add{than moving data into \pl{Postgres} and run the query there}.
In particular, we observed that loading data into \pl{Postgres} is already approximately $3\times$ slower than it takes \rheem to complete the entire task.
Even when discarding data migration times, \rheem performs quite similarly to \pl{Postgres}.
This is because, as our optimizer chooses to run this task on Spark, \rheem can parallelize most part of the task execution, leading to performance speedup.
For example, the pure execution time in \pl{Postgres} for a scale factor of $100$ amounts to $1,541$ seconds compared to $1,608$ seconds for \rheem.
We also observe that our optimizer has negligible overhead over the case when the developer writes ad-hoc scripts to move the data to HDFS for running the task on \pl{Spark}.
In particular, \rheem is twice faster than \pl{Spark} for scale factor $1$ because it moves less data from \pl{Postgres} to \pl{Spark}.
This shows the {\em substantial benefits of our optimizer in polystore scenarios, not only in terms of performance but also in terms of ease-of-use, as users do not write ad-hoc scripts anymore to integrate different platforms.}

\begin{figure}[!t]
	\centering
	\includegraphics[scale=0.4]{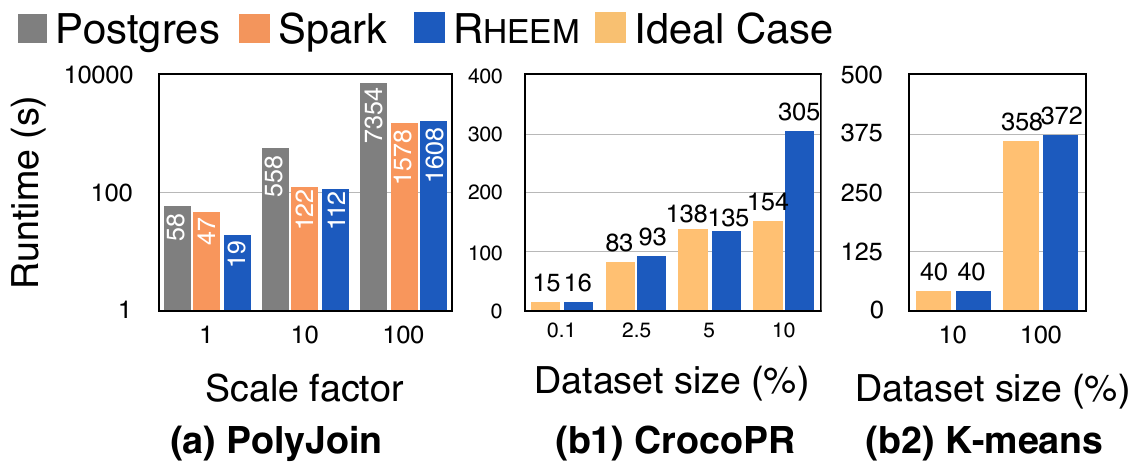}
	\vspace{-0.3cm}
	\caption{{\bf (a)}~Polystore \& {\bf (b)}~Mandatory multi-platform.}
	\label{fig:polymandatory}
	\vspace{-0.3cm}
\end{figure}


\myparagraph{Experiment setup (mandatory)}
To evaluate this feature, we consider the \task{CrocoPR} and \task{Kmeans} tasks.
In contrast to previous experiments, we assume both input datasets (\ds{DBpedia} and \ds{USCensus1990}) to be on \pl{Postgres}.
As the implementation of these tasks on \task{Postgres} would be very impractical and of utterly inferior performance, it is important to move the computation to a different processing platform.
In these experiments, we consider as baseline the \emph{ideal case} where the data resides in the HDFS instead and \rheem uses either \pl{JavaStreams} or \pl{Spark} to run the tasks.

\myparagraph{Results (mandatory)}
Figure~\ref{fig:polymandatory}(b) shows the results.
We observe that \rheem achieves similar performance with the ideal case in almost all scenarios.
This is a remarkable result, as it needs to move data out of \pl{Postgres} to a different processing platform, in contrast to the ideal case.
Only for \task{CrocoPR} and only for the largest dataset we measured a slow-down of \rheem \wrt~the ideal case, which is because \rheem reads data in parallel in the ideal case (\ie~when reading from \pl{HDFS}), which is not possible when reading from \pl{Postgres}.
Nevertheless, it is more efficient (and practical) than writing ad-hoc scripts to move data out of \pl{Postgres} and running the task on a different platform.
In particular, we observed that the optimizer \add{chooses to perform projections and selections in \textsf{Postgres}}, and thus reduces the amount of data to be moved.
These results show that {\em our optimizer frees users from the burden of complementing the functionalities of diverse platforms, without sacrificing performance}.

\subsection{Optimizer Scalability}
\label{section:experiments_indepth}

We continue our experimental study by evaluating the scalability of our optimizer in order to determine whether it operates efficiently on large \rheem plans and for large numbers of platforms.

\myparagraph{Experiment setup}
We start by evaluating our optimizer's scalability in terms of number of supported platforms and then proceed to evaluate it in terms of number of operators in a \rheem plan.
For the former, we consider hypothetical platforms with full \rheem operator coverage and three communication channels each.
For the latter, we generated \rheem plans with three basic topologies that we found to be at the core of many data analytic tasks: {\em pipeline}, {\em fanout}, and {\em tree}.
Notice that most iterative analytics also fall into these three topologies.

\myparagraph{Results}
Figure~\ref{figure:platformscalability} shows the optimization time of our optimizer for \task{Kmeans} when increasing the number of supported platforms -- the results for the other tasks are similar.
As expected, the time increases along with the number of platforms.
This is because (i)~the CCG gets larger, challenging our MCT algorithm, and (ii)~our lossless pruning has to retain more alternative subplans.
Still, we observe that our optimizer (the \emph{no top-$k$} series in Figure~\ref{figure:platformscalability}) performs well for a practical number of platforms: it takes less than $10$ seconds when having $5$ different platforms.
Yet, when adding a simple top-$k$ pruning strategy, our optimizer gracefully scales with the number platforms, \eg~for $k{=}8$ it takes less than $10$ seconds when having $10$ different platforms.
Note that our algebraic formulation of the plan enumeration problem allows to easily augment our optimizer with a top-$k$ pruning strategy (see Section~\ref{sec:enumeration_algebra}):
We just specify an additional rule for the \emph{Prune} operator.
%
Let us now proceed to evaluate our optimizer's scalability \wrt~the number of operators in a task.
Figure~\ref{figure:enumerationscalability} depicts the above mentioned plan topologies along with our experimental results.
The optimizer scales to very large plans for the pipeline and tree topologies.
In contrast, we get a different picture for the fanout topology:
The optimizer processed plans with at most 12~operators\footnote{The tasks from Table~\ref{tab:datasets} have 17~operators on average.} within a time frame of $5$ minutes.
Such plans entail hard MCT problems and allow for only very little (lossless) pruning.
However, we encountered only much smaller fanouts in real-world tasks.
We can thus conclude that {\em our optimizer can scale to a realistic number of platforms and to a reasonable number of operators in a \rheem plan.}


\begin{figure}[t!]
	\centering
	\subfigure[Scalability w.r.t platforms]{\includegraphics[scale=.25]{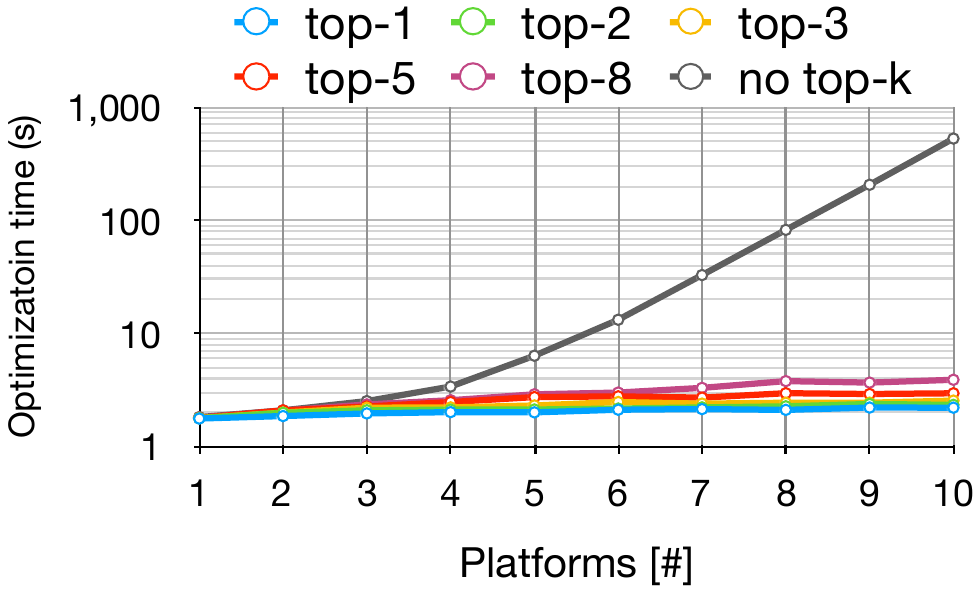}\label{figure:platformscalability}}
	\subfigure[Enumeration scalability]{\includegraphics[scale=.35]{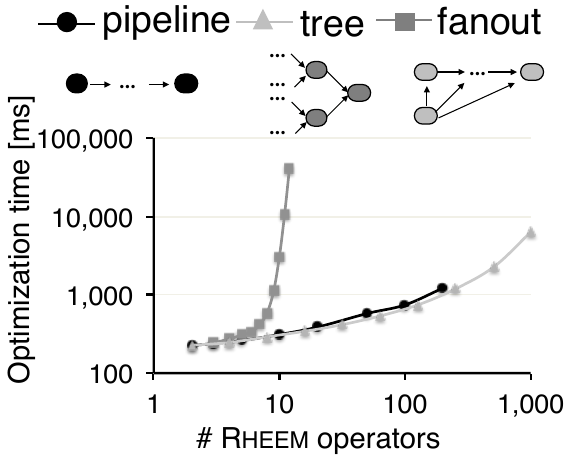}\label{figure:enumerationscalability}}
	\vspace{-0.4cm}
	\caption{Optimization scalability.}
	\label{figure:indepth}
	\vspace{-0.3cm}
\end{figure}

\subsection{Optimizer In Depth}
\label{section:extraexps}

Besides the scalability results, we also conducted several experiments to further evaluate the efficiency of our optimizer.
We start by analyzing the importance of the order, in which our enumeration algorithm process join groups (see Section~\ref{sec:enumeration_algorithm}).
As we observe in Figure~\ref{figure:extrasinternals}(a) for the tree topology, ordering the join groups order \emph{can} indeed be crucial.
On the other hand, the process of ordering the join groups does not seem to exert any measurable influence on the optimization time.

\begin{figure}[t!]
	\centering
	\includegraphics[width=\linewidth]{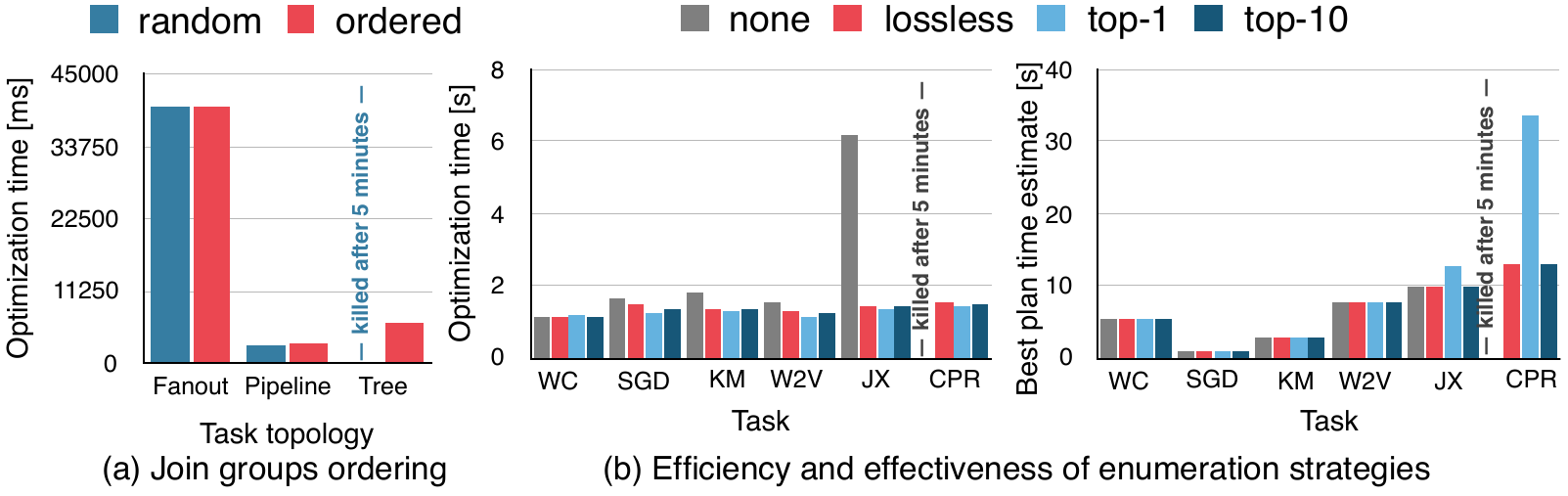}
	\caption{Optimizer internals.}
	\label{figure:extrasinternals}
\end{figure}

Additionally, we compare our lossless pruning strategy (Section~\ref{section:enumeration}) with several alternatives, namely no pruning at all and top-$k$ pruning%
\footnote{This is the same pruning as in Section~\ref{section:experiments_indepth}. However, while in Section~\ref{section:experiments_indepth} we used top-$k$ pruning to \emph{augment} our lossless pruning, here we consider it \emph{independently}.}
that retains the $k$ best subplans when applied to an enumeration.
Figure~\ref{figure:extrasinternals}(b) shows the efficiency results of all pruning strategies (on the left) as well as their effectiveness (on the right), \ie~the estimated execution times of their optimized plans.
Note that we did not use the actual plan execution times to assess the effectiveness of our enumeration strategy in order to eliminate the influence of the calibration of the cost functions.
As a first observation, we see that pruning is crucial overall:
An exhaustive enumeration was not possible for \task{CrocoPR} (\task{CPR}).
On the other hand, we found that the top-1 strategy, which merely selects the best alternative for each inflated operator, is pruning too aggressively and fails ($3$ out of $7$ times) to detect the optimal execution plan.
While the numbers now seem to suggest that the remaining lossless and top-10 pruning strategies are of the same value,
there is a subtle difference, though:
The lossless strategy \emph{guarantees} to find the optimal plan (w.r.t. the cost estimates) and is, thus, superior.
For large, complex \rheem plans, as discussed in the above paragraph, a combination of the lossless pruning followed by a top-$k$ pruning might be a valuable pruning strategy.
While the former keeps intermediate subplans diverse, the latter removes the worst plans.
This flexibility is a direct consequence of our algebraic approach to the plan enumeration problem.

\begin{figure}[t!]
	\centering
	\includegraphics[scale=0.3]{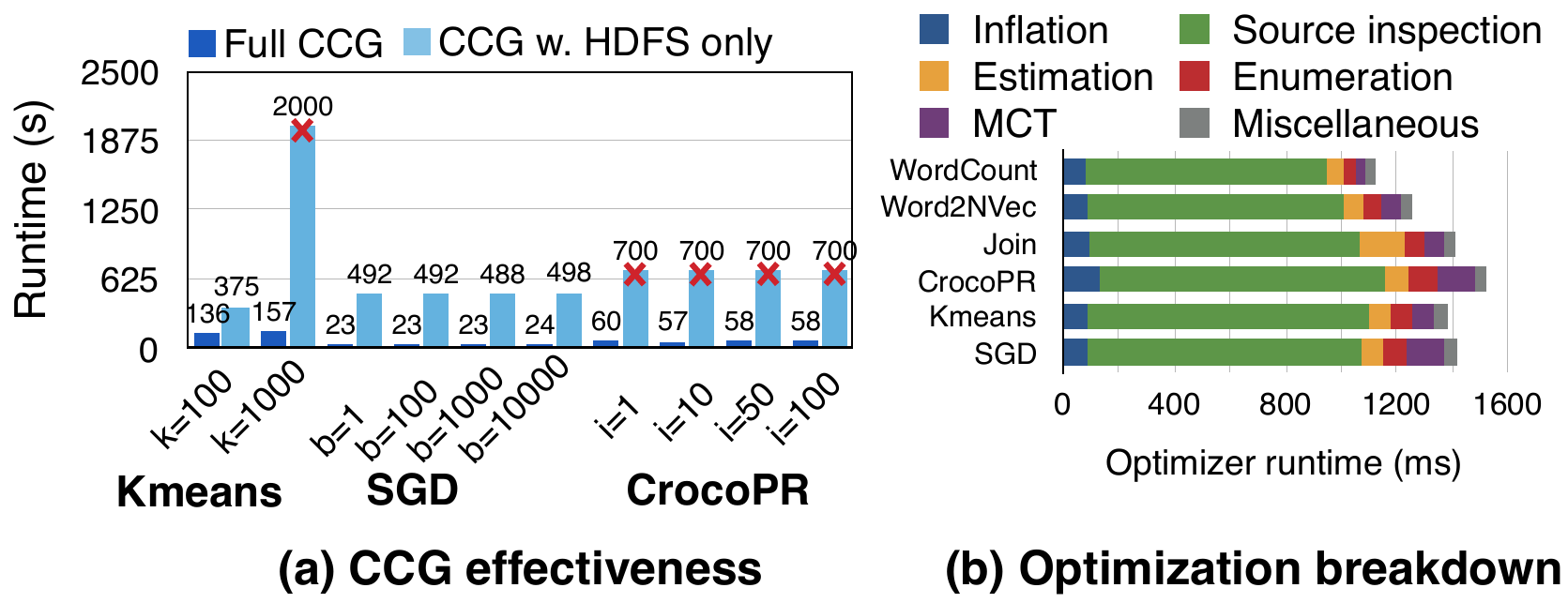}
	\caption{\add{CCG and optimization time breakdown.}}
	\label{figure:ccg}
\end{figure}

We additionally evaluate the effectiveness of \ccg.
For this, we assume that data movement is done only through writing to an HDFS file and thus, kept only the \at{HDFSFile} channel in \ccg and disabled all the rest.
Figure~\ref{figure:ccg}(a) shows the results in terms of runtime.
We observe that for \task{k-means} \rheem with the full \ccg is up to more than one order of magnitude faster than using only an HDFS file for the communication.
For \task{SGD} and \task{CrocoPR}, it is always more than one order of magnitude faster, in fact for \task{CrocoPR}, we had to kill the process after $700$ seconds.
This shows the high importance of our approach for the data communication.

\begin{figure}[t]
	\centering
	\includegraphics[scale=0.28]{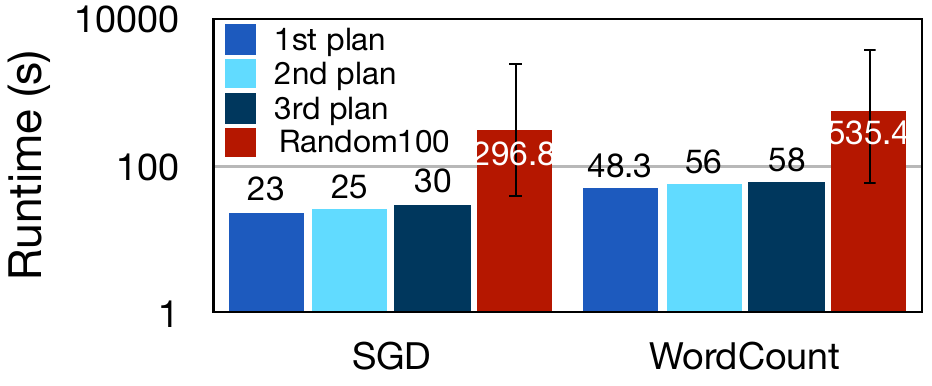}
	\caption{Cost model accuracy.	\label{figure:cost-model}}
\end{figure}

	We now validate the accuracy of our cost model.
	Note that similarly to traditional database cost-based optimizers, our cost model aims at enabling the optimizer to choose a good plan avoiding worst cases.
	That is, it does not aim at precisely estimating the running time of each plan.
	Thus, we evaluate the accuracy of our cost model by determining which plan of the search space our optimizer chooses.
	The ideal case would be to exhaustively run all possible execution plans and validate that our optimizer chooses (close to) the best plan.
	However, running all plans is infeasible as it would take weeks to complete only the small task of \task{WordCount} consisting of only $6$ operators.
	For this reason, in Figure~\ref{figure:cost-model} we plot for \task{SGD} and \task{WordCount} the following:
	(i)~the {\em real} execution time of the first three plans with the minimum {\em estimated} runtime (1st plan, 2nd plan, and 3rd plan); and
	(ii)~the minimum, maximum, and average of the real execution times of $100$ randomly chosen plans (Random100).
	We make the following observations.
	First, the 1st plan has the minimum real execution time compared to all other plans (including the 2nd and 3rd plans).
	Second, the first three plans have a better runtime not only compared to the average real execution time of the randomly chosen plans, but also compared to the minimum execution time of the randomly chosen plans.
	Based on these observations, we conclude that our cost model is sufficient for our optimizer to choose a near-optimal plan.

Finally, we analyze where the time is spent throughout the entire optimization process.
Figure~\ref{figure:ccg}(b) shows the breakdown of our optimizer's runtime in its several phases for several tasks.
At first, we note that the average optimization time amounts to slightly more than a second, which is several orders of magnitude smaller than the time savings from the previous experiments.
The lion's share of the runtime is the source inspection, which obtains cardinality estimates for the source operators of a \rheem plan (\eg~for inspecting an input file).
This could be improved, \eg~by a metadata repository or caches. In contrast, the enumeration and MCT discovery finished in the order of tens of milliseconds, even though they are of exponential complexity.


\section{Related Work}
\label{section:relatedwork}

In the past years, the research and industry communities have proposed many data processing platforms~\cite{MapReduce-ACM,apacheSpark,PostgreSQL,husky,stratosphere}.
In contrast to all these works, we do not provide a new processing platform but an optimizer to automatically combine and choose among such different platforms.

\noindent{\bf Cross-platform task processing} has been in the spotlight very recently.
Some works have proposed different solutions to decouple data processing pipelines from the underlying platforms~\cite{gog2015musketeer,bigdawg-demo,myria-cidr17,ires-bigdata,SimitsisWCD12,dbms+,tensorflow2015-whitepaper}.
Although their goals are similar, all these works differ substantially from our optimizer, as most of them do not consider data movement costs, which is crucial in cross-platform settings.
Note that some complementary works~\cite{weld,pipegen} focus on improving data movement among different platforms, but they do not provide a cross-platform optimizer.
Moreover, each of these systems \emph{additionally} differs from our optimizer in various ways.
Musketeer's main goal is to decouple query languages from execution platforms~\cite{gog2015musketeer}.
Its main focus lies on converting queries via a fixed intermediate representation and thus mostly targets platform independence.
BigDAWG~\cite{bigdawg-demo} comes with no optimizer and requires users to specify where to run cross-platform queries via its \texttt{Scope} and \texttt{Cast} commands.
Myria~\cite{myria-cidr17} provides a rule-based optimizer which is hard to maintain as the number of underlying platforms increases.
In~\cite{ires-bigdata} the authors present a cross-platform system intended for optimizing complex pipelines.
It allows only for simple one-to-one operator mappings and does not consider optimization at the atomic operator granularity.
The authors in~\cite{SimitsisWCD12} focus on ETL workloads making it hard to extend their proposed solution with new operators and other analytic tasks.
DBMS+~\cite{dbms+} is limited by the expressiveness of its declarative language and hence it is neither adaptive nor extensible.
Furthermore, it is unclear how DBMS+ abstracts underlying platforms seamlessly.
Other works, such as~\cite{pipegen,weld}, focus on improving data movement among different platforms and are complementary to our work.
Apache Calcite~\cite{apacheCalcite} decouples the optimization process from the underlying processing making it suitable for integrating several platforms. However, no cross-platform optimization is provided.
Tensorflow~\cite{tensorflow2015-whitepaper} follows a similar idea but for cross-device execution of machine learning tasks and thus it is orthogonal to \rheem.

\noindent{\bf Query optimization} has been the focus of a great amount of literature~\cite{queryoptimization}.
However, most of these works focus on relational-style query optimization, such as operator re-ordering and selectivity estimation, and cannot be directly applied to our system.
More closely to our work is the optimization for federated DBMSs where adaptive query processing and re-optimization is of great importance~\cite{BabuB05,rio,markl04}.
Nevertheless, the solutions of such works are tailored for relational algebra and assume tight control over the execution engine,
which is not applicable to our case.
Finally, there is a body of work on UDF-based data flow optimization, such as~\cite{sofa,hueske2012opening}.
Such optimizations are complementary to our optimizer and one could leverage them to better incorporate UDFs in our cost models.

\noindent{\bf MapReduce-based integration systems}, such as~\cite{miso,polybase},
mainly aim at integrating Hadoop with RDBMS and cannot be easily extended to deal with more diverse data analytic tasks and different processing platforms.
There are also works that automatically decide whether to run a MapReduce job locally or in a cluster, such as FlumeJava~\cite{chambers2010flumejava}.
Although such an automatic choice is crucial for some tasks, it does not generalize to data flows with other platforms.

\noindent{}Finally, {\bf federated databases} have been studied since almost the beginnings of the database field itself~\cite{federatedDBsSurvey}.
Garlic~\cite{garlic}, TSIMMIS~\cite{tsimmis}, and InterBase~\cite{interbase} are just three examples.
However, all these works significantly differ from ours in that they consider a single data model and push query processing to where the data is.

\section{Conclusion}
\label{section:conclusion}

We presented a cross-platform optimizer that automatically allocates a task to a combination of data processing platforms in order to minimize its execution cost.
Our optimizer considers the nature of the incoming task, platform characteristics, and data movement costs in order to select the most efficient platforms for a given task.
In particular, we proposed (i)~novel strategies to map platform-agnostic tasks to concrete execution strategies;
(ii)~a new graph-based approach to plan data movement among platforms;
(iii)~an algebraic formalization and novel solution to select the optimal execution strategy;
and~(iv)~how to handle the uncertainty found in cross-platform settings.
Our extensive evaluation showed that our optimizer allows tasks to run up to more than one order of magnitude faster than on any single platform.



\bibliographystyle{abbrv}
\bibliography{rheem}

\balance

\end{document}